\documentclass[reqno, 11pt]{amsart}
\usepackage{amssymb}
\usepackage[numbers]{natbib}
\pdfoutput=1

\usepackage[nesting]{hyperref}

\usepackage[pdftex]{graphicx}
\usepackage{listings}
\usepackage{multirow}
\usepackage{placeins}
\usepackage{color}
\usepackage{subfigure}
\usepackage{lscape}
\usepackage{tikz} 
\usetikzlibrary{patterns,snakes,arrows}
\usepackage{xcolor}
\usepackage{amsmath}
\usepackage[numbers]{natbib}

\usepackage{epstopdf}

\newcommand{\BlackBoxes}{\global\overfullrule5pt}

\BlackBoxes

\newcommand{\R}{\mathbb{R}}
\newcommand{\N}{\mathbb{N}}

\newcommand{\E}{\mathbb{E}}

\newcommand{\bS}{\mathbb{S}}

\newcommand{\cL}{\mathcal{L}}

\renewcommand{\P}{\mathbb{P}}

\newtheorem{theorem}{Theorem}
\newtheorem{corollary}[theorem]{Corollary}

\newtheorem{proposition}[theorem]{Proposition}
\theoremstyle{definition}
\newtheorem{example}[theorem]{Example}
\newtheorem{remark}[theorem]{Remark}

\numberwithin{equation}{section} \numberwithin{theorem}{section}

\def\0{\kern0pt\-\nobreak\hskip0pt\relax}

\makeatletter
\AtBeginDocument{%
 \def\@serieslogo{%
 \vbox to\headheight{%
 \parindent\z@ \fontsize{6}{7\p@}\selectfont
 \vss}}}

\def\makeoverbar#1#2#3#4#5#6#7{%
 \setbox0=\hbox{$\m@th#2\mkern#5mu{{}#3{}}\mkern#6mu$}%
 \setbox1=\null \dimen@=#4\fontdimen8#13 \dimen@=3.5\dimen@
 \advance\dimen@ by \ht0 \dimen@=-#7\dimen@ \advance\dimen@ by \wd0
 \ht1=\ht0 \dp1=\dp0 \wd1=\dimen@
 \dimen@=\fontdimen8#13 \fontdimen8#13=#4\fontdimen8#13
 \rlap{\hbox to \wd0{$\m@th\hss#2{\overline{\box1}}\mkern#5mu$}}
 \fontdimen8#13=\dimen@}

\def\mylabel#1#2{{\def\@currentlabel{#2}\label{#1}}}

\makeatother

\newcommand{\ind}{\textbf{1}}
\newcommand{\ma}{\text{max}}
\newcommand{\mi}{\text{min}}

\newcommand{\ol}{\overline}
\newcommand{\ul}{\underline}

\newcommand{\Rp}{\mathbb{R}^+}
\newcommand{\cP}{\mathcal{P}}
\newcommand{\cB}{\mathcal{B}}

\newcommand{\eps}{\varepsilon}
\newcommand{\Ydt}{Y_{\Delta t}}
\newcommand{\qdt}{q_{\Delta t}}
\newcommand{\dt}{\Delta t}

\newcommand{\up}{\uparrow}
\newcommand{\down}{\downarrow}

\newcommand{\vphi}{\varphi}

\begin{document}


\makeatletter \providecommand\@dotsep{5} \makeatother

\title[Gas Storage valuation with regime switching]{Gas storage valuation with regime switching}

\author[N. \smash{B\"auerle}]{Nicole B\"auerle${}^*$}
\address[N. B\"auerle]{Institute for Stochastics,
Karlsruhe Institute of Technology, D-76128 Karlsruhe, Germany}

\email{\href{mailto:nicole.baeuerle@kit.edu}
{nicole.baeuerle@kit.edu}}

\author[V. \smash{Riess}]{Viola Riess${}^*$}
\address[V. Riess]{Institute for Stochastics,
Karlsruhe Institute of Technology, D-76128 Karlsruhe, Germany}

\email{\href{mailto:viola.riess@kit.edu} {viola.riess@kit.edu}}

\thanks{${}^*$ Department of Mathematics,
Karlsruhe Institute of Technology, D-76128 Karlsruhe, Germany. Corresponding author: nicole.baeuerle@kit.edu, Tel.: +49-721-608-48152, Fax: +49-721-608-46066 }

\maketitle

\begin{abstract}
In this paper we treat a gas storage valuation problem as a Markov Decision Process. As opposed to existing literature we model
the gas price process as a regime-switching model. Such  a model has shown to fit market data quite well in \cite{chen2010implications}.
Before we apply a numerical algorithm to solve the problem, we first identify the structure of the optimal injection and withdraw policy.
This part extends results in \cite{Sec10}. Knowing the structure reduces the complexity of the involved recursion in the algorithms by
one variable. We explain the usage and implementation of two algorithms:  A Multinomial-Tree Algorithm and a Least-Square Monte Carlo
Algorithm. Both algorithms are shown to work for the regime-switching extension. In a numerical study we compare these two algorithms.
\end{abstract}

\vspace{0.5cm}
\begin{minipage}{14cm}
{\small
\begin{description}
\item[\rm \textsc{ Key words}]
{\small Gas Storage Valuation, Regime Switching, Markov Decision Process, Least-Square Algorithm, Multinomial-Tree Algorithm.}
\end{description}
}
\end{minipage}

\section{Introduction}\label{sec:intro}\noindent
In the EU, natural gas is still the second fuel used in final energy
consumption with a share of $22\%$, ahead of electricity at
$20\%$ and behind oil products (cp. \cite{eurogas}). A lot of factors
influence the gas price like weather conditions, economic growth, prices
of other fuel like oil, coal, carbon, growing share of renewables and
new production techniques for unconventional gas resources. In order to
handle the resulting gas price volatility, the gas storage capacity has
been increased constantly. The number of storage facilities in the EU in
year 2013 is 142 with a maximum working volume of 96597 million cubic
meters (cp. \cite{eurogas}). This storage on one hand can be used to balance
supply and demand but also on the other hand to create profit from an
active storage management on a mark-to-market basis. Contracts for
natural gas storage essentially represent real options on natural gas
prices. We refer to \cite{geman2005commodities} for an introduction to
the storage valuation problem.

In order to evaluate a contract on a gas storage facility we need a
stochastic model for the gas price process and a tool to find the
optimal dynamic injection and withdraw policy. The latter problem
results in a stochastic control problem and there are quite some papers
in the literature which deal with it. Most often, the gas price models
are continuous and somehow related to interest rate models. We thus get
a continuous stochastic control problem which can be solved via the
Hamilton-Jacobi-Bellman equation. However, this equation has to be
solved numerically (see e.g.
\cite{chen2007semi,thompson2009natural,chen2010implications}). An
alternative approach is to discretize the gas price process first and
then treat the resulting discrete-time problem as a Markov Decision
Process (see e.g.
\cite{BdJ08,Sec10,lai2010approximate,boogert2011gas,felix2012gas}). An
intermediate model has been considered in \cite{carmona2010valuation}
where the underlying process is continuous, however the operation
policies are restricted to switching policies where the storage level
switches between full and empty. The task can then be reduced to the
discrete problem of finding the optimal stopping times for the switches
of operation modes.  Surprisingly there is only the paper by
\cite{Sec10} which deals with obtaining the analytic structure of the
optimal policy. It is crucial there, but also in other papers, that
besides the natural upper and lower capacity bound of the storage, there
is a restriction on the maximum rate at which gas can be released from
the storage and a maximum rate at which gas can be injected into the
storage. These rates depend on the storage level and are typically
decreasing and convex or concave respectively. In \cite{Sec10} these
rates are assumed to be constant. In any case, as a result the optimal action not
only depends on the gas price but also on the storage level. The optimal
policy can be characterized by three regions which depend on the gas
price $p$: When the current gas storage level is below a certain bound
$\underline{b}(p)$, it is optimal to inject gas, either as much as
possible or up to $\underline{b}(p)$ which occurs first. If the current
gas storage level is above a certain bound $\bar{b}(p)$, it is optimal
to withdraw gas, either as much as possible or down to $\bar{b}(p)$
which occurs first. In case the level is in between, it is optimal to do
nothing.

In this paper we extend the result of \cite{Sec10} to more general
restrictions on the maximum injection and withdraw rate and to more
general gas price processes. Moreover we use the structure of the
optimal policy to implement efficient numerical algorithms for the
solution of the gas storage problem. More precisely we use a
regime-switching model and thus a non-Markovian gas price model. The
model is based on the mean-reverting model of \cite{Schwartz97}. In our
setting, there are two regimes which result in different time-dependent,
seasonal mean-reversion factors. It has been shown in \cite{Schwartz97,
jaillet2004valuation,chen2010implications} that a one-factor mean-reverting model does not
capture the dynamics of a typical gas forward curve. Hence there is on
the one hand the possibility to use multi-factor models, like in
\cite{cartea2008uk,boogert2011gas,lai2010approximate,mhm} or to introduce regime-switches,
see e.g. \cite{chen2010implications}. It is argued in \cite{chen2010implications} that
regime-switching models calibrate quite well to market data and are less
complex than multi-factor models. For a further discussion of gas price
models see \cite{eydeland2003energy}. Our gas price processes have no
jumps (for jump models see \cite{chen2007semi,chen2010implications}) but these could be
included easily.

As already mentioned, we use the structure of the optimal policy to
implement efficient numerical algorithms for the solution of the gas
storage problem. The basic structure of the algorithms is given by the
usual backward induction algorithm from the theory of Markov Decision Processes.
This algorithm suffers from the
curse of dimensionality. However, when implemented in a clever way, this
effect can be mitigated  (see e.g. \cite{powell2007approximate} for approximate dynamic
programming). We will see that
the usage of the structure of the optimal policy reduces the number of
optimization problems which have to be solved by one variable. The two
algorithms we compare, differ in the way the expectation in the recursion
is computed. In the first algorithm we approximate the continuous gas
price process by a recombining tree, using the method proposed in
\cite{Nel90}. This method has the advantage that the values of the tree
are unaffected by the regime which only influences the transition probabilities.
Other methods to construct trees are explained in
\cite{jaillet2004valuation,felix2012gas,bardou2009optimal}. However, these
methods do not extend easily to the regime-switching model. In the second algorithm we use
the Least-Square Monte Carlo Algorithm to approximate the expectation by
regression on a number of basis functions. This algorithm is well-known
from \cite{LS01} for the valuation of American options and first used for gas storage valuation by
\cite{BdJ08}. In \cite{boogert2011gas} the authors use the
Least-Square Monte Carlo Algorithm for a multi-factor gas price model
and discuss the impact of the choice of the basis functions. Here we show that this
method can also be applied to regime-switching models.

Our paper is organized as follows: In the next section we formulate the
gas storage valuation problem with regime switching gas price as a
Markov Decision Process. Afterwards we prove in Section \ref{sec:structure}
the structure of the optimal policy and consider some special cases.
Next in Section \ref{sec:algorithms} we specify our gas price model and
explain the Multinomial-Tree Algorithm and the Least-Square Monte Carlo
Algorithm. The parameters of the gas price model are taken from \cite{BBK08}.
In this section we also discuss an efficient discretization of the gas storage.
A numerical analysis of the two algorithms can be
found in Section \ref{ssec:numeric}, where we use a specific example. We discuss the choice
of the number of grid points for the gas storage level, the number of simulations
necessary for the Least-Square Monte Carlo Algorithm and the number of nodes for the
tree in the Multinomial-Tree Algorithm.

\section{The Gas Storage Valuation Problem as a Markov Decision Process with Regime Switching}\label{sec:model}\noindent
We formulate the gas storage valuation problem as a Markov Decision Process (MDP) with regime switching.  For the theory of Markov Decision Processes we
refer to \cite{BR11,ber02,herlerlas96,puterman}. Let $(P_n)$ be the stochastic gas price process which is defined on the measurable space $(\cP, \cB(\cP))$, where $\cP \subseteq \Rp$. This price process is influenced by a Markov chain $(R_n)$, such that $(P_n, R_n)$ is a Markov process. We further assume that $(R_n)$ has finite state space $S_R$ and transition matrix $(q_{jk})_{j,k\in S_R}$. Additionally assume that $P_{n+1}$ and $R_{n+1}$ are conditionally independent under $R_n$ and $P_n$ and that $\cL(R_{n+1}|P_n, R_n)= \cL(R_{n+1}|R_n)$, where $\cL$ denotes the law of the corresponding random variables. Using these assumptions we get for $j, k \in S_R$, $p \in \cP$, $B\in \cB(\cP)$
\begin{align*}
&\P(P_{n+1} \in B , R_{n+1}=k | R_n=j, P_n=p)\\
 & = \P(P_{n+1} \in B| R_n=j, P_n=p)\P( R_{n+1}=k | R_n=j, P_n=p)\\
& = \P(P_{n+1} \in B| R_n=j, P_n=p)\P( R_{n+1}=k | R_n=j)\\
&=: Q_{n,j}(B|p) \cdot q_{jk}.
\end{align*}
Let $N \in \N$ be the finite planning horizon and denote by $b^\mi$ and $b^\ma$ the minimal and maximal capacity  of the gas storage facility, respectively. The state space of the Markov Decision Process is given by $E= [b^\mi, b^\ma]\times \cP\times S_R$ and we denote the elements of $E$ by $(x,p,r)\in E$ where $x$ is the current storage level, $p$ the current price and $r$ the current background state. The action space is given by $A= [b^\mi-b^\ma, b^\ma-b^\mi]$, where  $a$ is the change in the amount of gas. Hence  $a<0$ means that the amount of $|a|$ is withdrawn and $a>0$ means that $|a|$ is injected. Since we assume that actions are additionally restricted by maximal rates of changes $i_n^\mi$ and $i_n^\ma$ that depend on the current amount of gas in the storage, the admissible actions are restricted to the set $D_n(x)= \{a \in A | \max(b^\mi-x, i_n^\mi(x)) \le a \le \min(b^\ma -x, i_n^\ma(x))\}$. We assume that $i_n^\ma$ and $i_n^\mi$ are decreasing functions of the amount of gas in storage $x$. Furthermore $i_n^\ma$ is supposed to be concave and non-negative and $i_n^\mi$ is supposed to be convex and non-positive in the interval $[b^\mi, b^\ma].$ The set of admissible actions is illustrated in Figure \ref{fig:restriction}. A decision rule at time $n$ is a measurable mapping $f_n : E \to D_n$ and a policy $\pi=(f_0,\ldots,f_{N-1})$ is defined by a sequence of decision rules.\\

\begin{figure}[htbp]
\centering
\begin{tikzpicture}

\shadedraw[domain= 1:3 , color=black!10, variable= \x, smooth, left color=black!10,right color=black!10] (1,0) -- plot(\x, {min(sqrt(-0.2*\x+0.7), (3-\x))}) --  (1,0);

\shadedraw[domain= 1:3 ,  color=black!10,  variable= \x, smooth, left color=black!10,right color=black!10] (1,0) -- plot(\x, {max(1-\x,-sqrt(0.4*\x))}) -- (3,0);

\draw [->] (-0.5,0) -- (3.5,0);
\draw [->]  (0, -2.5) -- (0,3.5);

\node  [left, color=black] at (-0.1,1) {\small{$b^\mi$}};
\node  [left, color=black] at (-0.1,3) {\small{$b^\ma$}};

\node  [right, color= black] at (1.1,3.5) {\small{$b^\mi$}};
\node  [right, color= black] at (3.1,3.5) {\small{$b^\ma$}};

\node  [below, color= black] at (0.5,-0.5) {\small{$i^\mi_n$}};
\node  [above, color= black] at (3.5,0.1) {\small{$i^\ma_n$}};

\node  [left] at (-0.1,3.5) {\small{$a$}};

\node  [left] at (3.5,-0.2) {\small{$x$}};

\draw[-] (-0.05,1)--(0.05,1);
\draw[-] (-0.05,3)--(0.05,3);

\draw[dashed, color=black] (0,1) -- (3,-2);
\draw[dashed, color=black] (0,3) -- (3,0);

\draw[dashed, color= black] (1,3.5) -- (1,-2.5);
\draw[dashed, color= black] (3,3.5) -- (3,-2.5);

\draw[domain=1:3.25, smooth, dashed, color= black] plot (\x,{sqrt(-0.2*\x+0.7)}) ;

\draw[domain=0.5:3, smooth, dashed, color= black] plot (\x,{-sqrt(0.4*\x)}) ;

\node  [color=black] at (2.5,-0.4) {\small{$D_n$}};

\end{tikzpicture}
\caption[Restriktionenmenge]{Set of admissible actions $D_n$}\label{fig:restriction}
\end{figure}
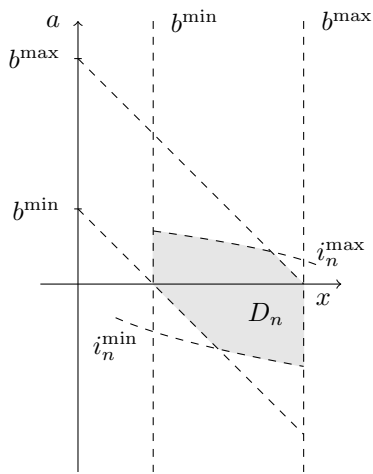

In order to include transaction costs, a loss of gas at the pump and/or a bid ask spread in the market, we introduce the ``ask price'' $k$ and the ``bid price'' $e$ for one quantity of gas. In particular we assume that $k(p)= (1+ w_1) p + z_1$ and $e(p) = (1- w_2)p - z_2$, where $w_1, w_2, z_1, z_2 \in \R_0^+$ such that $k(p) \ge e(p) \ge 0$ for each $p \in \cP$.  Thus we have fixed and proportional transaction cost. The one stage reward function $h$ of the Markov Decision Process is given by\\
\begin{align*}
h(p,a) = \begin{cases}
- k(p) \cdot a, & a>0,\\
0, & a=0 \\
-e(p) \cdot a, & a < 0.
\end{cases}
\end{align*}
Let $h_N$ denote the terminal reward function, that depends on the current amount of gas in the storage $x$ and the current price $p$. We assume that $h_N$ is a concave function of $x$.

\begin{example}[Some possible terminal reward functions]
One possible terminal reward function could be $h_N(x,p)= e(p) (x- b^\mi)$. Here the terminal reward corresponds to selling the gas that is available in the storage at time $N$. An alternative choice for a terminal reward function that includes a penalty when the amount of gas in the storage is below some amount $x_{\text{end}}$ is
\begin{align*}
h_N(x,p)=  \begin{cases}
e(p) (x- x_{\text{end}}), & x \ge x_{\text{end}},\\
-k(p) (x_{\text{end}}-x), & x \le x_{\text{end}}.
\end{cases}
\end{align*}
In order not to favour situations where the amount of gas in the storage is above the level $x_{\text{end}}$, we can replace $e(p) (x- x_{\text{end}})$ by $0$ for $x \ge x_{\text{end}}$. Another modification is to let the ``penalty'' grow exponentially to the distance to $x_{\text{end}}$ and not proportionally.
\end{example}

The transition kernel for the Markov Decision Process is given by
\begin{align}
\widetilde Q_n(d(x',p',j)|x,p,r,a)= Q_{n, r}(d p'|p) q_{r j} \otimes \delta_{x+a}(d x').
\end{align}
By definition of the transition kernel we get for any measurable function $v$ and admissible $(x,p,r,a)$ that
\begin{align}
\int v(x', p',j) \widetilde Q_n(d(x',p',j)|x,p,r,a) = \sum_{j\in S_R} q_{r j} \int v(x+a, p',j) Q_{n,r}(d p'|p) .
\end{align}

Having introduced all the necessary definitions, the objective is to find an optimal policy $\pi$ that maximizes the expected total reward over time. Therefore we define the value function as the maximal expected accumulated reward from time $n$ up to time $N$, that is for $(x,p,r)$ $\in$ $E$
\begin{align}
V_n(x,p,r) = \sup_{\pi \in \Pi} \E^\pi_{n,x,p,r} \Big[\sum\limits_{k=n}^{N-1} \beta^k h(P_k,f_k(X_k,P_k,R_k)) + \beta^N h_N(X_N, P_N)\Big]
\end{align}
where $\Pi$ denotes all admissible policies and $\beta\in(0,1]$ is a discount factor. $\E_{n,x,p,r}[\ldots]$ denotes the conditional expectation given $X_n=x, P_n=p, R_n=r$. Using the Bellman equation we know, that we can solve the problem of finding $V_0$ by
\begin{align}
V_N(x,p,r) &= h_N(x,p)\\
\begin{split}\label{eq:bellman}
V_n(x,p,r) &= \sup\limits_{a \in D_n(x)} \{h(p,a) +\beta \E_{n,x,p,r}[V_{n+1}(x+a, P_{n+1}, R_{n+1})]\}\\
&= \sup\limits_{a \in D_n(x)} \Big\{h(p,a) + \beta\sum\limits_{j\in S_R} q_{rj} \int V_{n+1 }(x+a, p', j) Q_{n,r}(dp'|p)\Big\}
\end{split}
\end{align}

We introduce the following operators for some measurable function $v$
\begin{align}
(L_nv)(x,p,r,a)& = h(p,a) + \beta \sum\limits_{j\in S_R} q_{rj} \int v(x+a, p', j) Q_{n,r}(dp'|p),\\
(T_n v)(x,p,r) &= \sup\limits_{a \in D_n(x)}  (L_nv)(x,p,r,a),
\end{align}
such that equation (\ref{eq:bellman}) can be written as $V_n(x,p,r) = (T_nV_{n+1})(x,p,r)$.

\begin{remark}
It can be shown that if there exists $c_1,c_2,c_3\in \R_0^+$, such that
\begin{itemize}
\item $|h_N(x,p,r)| \le c_2 (k(p) + c_1),$
\item $\E_{n,p,r}[k(P_{n+1})+ c_1] \le c_3 (k(p) +c_1),$
\end{itemize}
the expected values from the definitions above exist. Thus the value function is well defined. In what follows we denote by
$$\bS_s^+ := \{ v:E\to \R : v^+(x,p,r) \le c_2(k(p)+c_1) \;\mbox{for arbitrary}\; c_1, c_2\in\Rp\}.$$
For details see the concept of a bounding function in \cite{BR11}.
\end{remark}


\section{Properties of the Optimal Gas Storage Policy and the Value Function}\label{sec:structure}
In this section we identify the structure of the optimal gas storage policy as explicitly as possible. It follows essentially from some properties of the value function which are shown by induction.
\subsection{Structure of the optimal policy for the general model}\label{ssec:structurepolicy}
\begin{proposition}\label{prop:concavity}
Let $v\in \bS_s^+$. If $x\mapsto v(x,p,r)$ is concave, then $x\mapsto T_nv(x,p,r)$ is concave on $\Rp$ for each fixed $p \in \cP$ and $r\in S_R.$
\end{proposition}

\begin{proof}
Let $x_1,x_2$ be admissible and $x= \alpha x_1+(1-\alpha) x_2$, $\alpha \in [0,1]$. Let $\eps >0$. Then there exist $a_1\in D_n(x_1)$ and $a_2\in D_n(x_2)$ such that
\begin{align*}
L_nv(x_1,p,r, a_1) \ge T_nv(x_1,p,r) -\eps, \\
L_nv(x_2,p,r, a_2) \ge T_nv(x_2,p,r) -\eps.
\end{align*}
Further we obviously have $\alpha a_1 + (1-\alpha) a_2 \in D_n(x)=D_n(\alpha x_1+(1-\alpha) x_2)$ by checking the restrictions.
Thus
\begin{align*}
 T_nv(x,p,r) &= &&T_nv(\alpha x_1+(1-\alpha) x_2,p,r)\\
&=&& \sup\limits_{a\in D_n(\alpha x_1+ (1-\alpha) x_2)}Lv(\alpha x_1+(1-\alpha) x_2,p,r,a)\\
&\ge&& Lv(\alpha x_1+(1-\alpha)x_2, p,r, \alpha a_1+(1-\alpha) a_2)\\
&= &&h(p,  \alpha a_1+(1-\alpha) a_2)\\
& && + \beta \sum\limits_{j\in S_R}q_{rj} \int v(\alpha x_1+(1-\alpha) x_2 + \alpha a_1+(1-\alpha) a_2,p',j) Q_{n,r}(dp'|p)\\
&
\ge && \alpha h(p, a_1) + (1-\alpha) h(p, a_2) \\
& && +\beta\sum\limits_{j\in S_R}q_{rj} \int \alpha v( x_1+a_1,p',j) + (1-\alpha) v( x_2+a_2,p',j) Q_{n,r}(dp'|p)\\
&= &&\alpha Lv(x_1,p,r,a_1) + (1-\alpha) Lv(x_2,p,r,a_2)\\
& \ge &&\alpha(T_nv( x_1,p,r)-\eps) + (1-\alpha)(T_nv(x_2,p,r)-\eps) \\
&= && \alpha  T_nv( x_1,p,r) + (1-\alpha) T_nv( x_2,p,r) - \eps,
\end{align*}
where the second inequality follows by assumption and the fact that $a \mapsto h(p,a)$ is concave. The assertion follows by taking $\eps\searrow 0.$
\end{proof}

For notational simplicity we now introduce the so called ``continuation value'', that is
\begin{align*}
U_N(x,p,r) &= 0\\
U_n(x,p,r) & = \beta \sum\limits_{j \in S_R} q_{rj} \int V_{n+1}(x, p', j) Q_{n,r}(dp'|p).
\end{align*}

From Proposition \ref{prop:concavity} and the fact that $V_N(\cdot, p, r) = h_N(\cdot, p)$ is concave and $U_N(\cdot, p, r)$ is constant, the next corollary can be deduced immediately.

\begin{corollary}\label{cor:concavity}
Both functions $V_n(\cdot, p,r)$ and $U_n(\cdot, p,r)$ are concave for all $p \in \cP, r \in S_R$ and $n=1,\ldots,N$.
\end{corollary}

\begin{theorem}[Structure of the optimal policy]\label{thm:optimalpolicy}
For each $n\in\{0,\ldots ,N-1 \}$ there exist $\ul b_n(p,r)$ and $\ol b_n(p,r)$ where $\ul b_n(p,r) \le \ol b_n(p,r)$, such that the optimal policy $\pi^\star =(f_1^\star,\ldots,f_{N-1}^\star)$ is given by
\begin{align}\label{equation:optimalpolicy}
f^\star_n(x,p,r) = \begin{cases}
\min(\ul b_n(p,r) -x, i_n^\ma(x)), & b^\mi \le x < \ul b_n(p,r),\\
0 , & \ul b_n(p,r) \le x \le \ol b_n(p,r),\\
\max(\ol b_n(p,r) -x, i_n^\mi(x)), & \ol b_n(p,r) < x \le b^\ma.
\end{cases}
\end{align}
\end{theorem}

\begin{proof}
The proof is similar to the proof of Theorem 1 in \cite{Sec10}. Let $n\in \{0,\ldots , N-1\}$ be fixed but arbitrary.
At first we loosen the restrictions of the optimization problem at time $n$ and maximize over all $z= a+x$ $\in \left[b^\mi,b^\ma\right]$. Thus we consider the following problem
\begin{align*}
\max\limits_{z\in\left[b^\mi,b^\ma\right] } h_n(z-x,p,r) + U_n(z,p,r),
\end{align*}
where $p \in \cP$ and $r\in S_R$. Dividing the optimization problem into the case $z\ge x$ and the case $z \le x$ yields
\begin{align}
\label{1}
\max \left(\max\limits_{z\in\left[x,b^\ma\right] }  U_n(z,p,r)-k(p) z+ k(p) x,
\max\limits_{z\in\left[b^\mi,x\right] }  U_n(z,p,r)-e(p) z+ e(p) x\right).
\end{align}
Let us denote the two subproblems by
\begin{align}
\label{2}
\max\limits_{z\in\left[x,b^\ma\right] }  U_n(z,p,r)-k(p) z+ k(p) x
\end{align}
\begin{align}
\label{3}
\max\limits_{z\in\left[b^\mi,x\right] }  U_n(z,p,r)-e(p) z+ e(p) x.\end{align}
Moreover, let  $\ul b_n(p,r)$ be the maximizer of \eqref{2} when $x=b^\mi$ and $\ol b_n(p,r)$ the maximizer of \eqref{3} when $x=b^\ma$.

In the following we will prove that $\ul b_n(p,r)  \le \ol b_n(p,r).$ For fixed $p,r$ let us denote
\begin{eqnarray*}
  g_1(z) &=& U_n(z,p,r) - k(p) z  \\
  g_2(z) &=& U_n(z,p,r) - e(p) z.
\end{eqnarray*}
The function $g$ with
$$g(z,p,r):= g_1(z)-g_2(z) = {(e(p) -k(p))} z$$
is decreasing in $z$ since $e(p) -k(p)<0$. Thus $g_1(z)$ is equal to $g_2(z)$ plus a decreasing function in $z$. Therefore the maximizer of $g_1$  is smaller or equal to the maximizer of $g_2$ for each $p,r$ fixed but arbitrary. Hence the assertion follows. \\

Consider the optimization problem \eqref{1} for $x\in[b^\mi, \ul b_n(p,r))$. Obviously, $\ul b_n(p,r) $ maximizes \eqref{2}. Now consider \eqref{3}. For each admissible $z$ we have $z\le x< \ul b_n(p,r) \le \ol b_n(p,r)$. Since $\ol b_n(p,r)$ maximizes the optimization problem \eqref{3} when $x=b^\ma$, but is not admissible, we know by the concavity of the function that $x$ maximizes the problem. Thus we get that $x$ maximizes \eqref{3}.\\
Now we need to determine whether $\ul b_n(p,r)$ or $x$ maximizes \eqref{1}.
Since $x$ is admissible for \eqref{2} and $\ul b_n(p,r)$ maximizes this problem we get
$$U_n(\ul b_n(p,r),p,r) - k(p) \ul b_n(p,r) + k(p) x \ge U_n(x,p,r) - k(p) x + k(p) x = U_n(x,p,r).$$
Inserting the maximizer $x$ in \eqref{3}, we get $U_n(x,p,r)$. Thus $\eqref{1}=\max(\eqref{2},\eqref{3})$ is maximized by $\ul b_n(p,r)$.\\

Considering the optimization for $x\in (\ol b_n(p,r),b^\ma]$ it can be shown by similar argumentation that $\ol b_n(p,r) $ maximizes \eqref{1}.\\

For $x\in [\ul b_n(p,r), \ol b_n(p,r) ]$ we have that $x$ maximizes both \eqref{2} and \eqref{3} and thus \eqref{1}. This is a similar argumentation as in the first case considering the optimization of \eqref{3}.\\

Now we restrict the admissible $z$ in \eqref{1}, such that $i_n^\mi(x) + x \le z\le i_n^\ma(x) + x$. Then we consider the problem
\begin{align}
\label{1stern}
\max\limits_{z\in \left[(i_n^\mi(x)+x) \vee b^\mi,\  b^\ma \wedge (i_n^\ma(x)+x)\right]}h_n(z-x,p,r) + U_n(z,p,r)
\end{align}
and get for the maximizers
\begin{align*}
z^\star(x,p,r) = \begin{cases}
\ul b_n(p,r) \wedge (i_n^\ma(x) + x), & x\in[b^\mi, \ul b_n(p,r))\\
x , & x\in[\ul b_n(p,r), \ol b_n(p,r) ], \\
\ol b_n(p,r) \vee (i_n^\mi(x) + x), & x\in (\ol b_n(p,r) , b^\ma].
\end{cases}
\end{align*}
Transferring this result to the original problem $(a= z-x)$ yields the statement.
\end{proof}

\begin{remark}
Under further assumptions it can be shown, that $\ul b_n(\cdot, r)$ and $\ol b_n(\cdot, r)$ are monotonically decreasing, see \cite{Riess14}.
\end{remark}

The interesting fact about the structure of the optimal policy is, that it is not always optimal to withdraw or inject as much gas as possible or do nothing. Instead by using the optimal policy you are trying to reach the calculated bound in the next step or do nothing depending on how much gas is in the storage. Figure \ref{fig:strategy} illustrates the structure of the optimal policy. The choice of the action depends on the amount of gas in the storage and the optimal policy ``divides'' the interval $[b^\mi, b^\ma]$ into three parts. Assuming the amount of gas in the storage $x_1$ lies between $b^\mi$ and $\ul b_n(p,r)$, then it is optimal to inject gas, either as much as possible ($|i_n^\ma(x_1)|$), if you do not or just reach the bound $\ul b_n(p,r)$, or just the amount that is missing to reach $\ul b_n(p,r)$ - although you could inject more. A similar situation arises when the amount of gas in the storage $x_2$ is above $\ol b_n(p,r)$.

\begin{figure}[htbp]
\begin{tikzpicture}
\begin{scope}[>=latex]
\draw[-] (0,0)-- node[above, midway]{}(3,0) -- node[above, midway]{}(6,0)--node[above, midway]{}(9,0);
\node at(1.5,0.6){\small{inject}};
\node at(4.5,0.6){\small{nothing}};
\node at(7.5,0.6){\small{withdraw}};

\draw[-] (0, 0.5)--(0,-0.5) node[below]{\textcolor{black}{\small{$b^\mi$}}};
\draw[-] (9, 0.5)--(9,-0.5) node[below]{\textcolor{black}{\small{$b^\ma$}}};

\draw[-] (3, 0.5)--(3,-0.5) node[below]{\small{$\ul b_n(p,r)$}};
\draw[-] (6, 0.5)--(6,-0.5) node[below]{\small{$\ol b_n(p,r)$}};

\draw[-, color=black] (0.5, 0.1)--(0.5,-0.1) node[below]{\textcolor{black}{\scriptsize{$x_1$}}};

\draw[->, color=black] (0.5,0).. controls (0.75,0.5) and (1.25,0.5).. (1.5,0);
\draw [color=black,
	decoration={
        brace,
        mirror,
        raise=0cm
    },
    decorate
] (0.5,-0.4) -- (1.5,-0.4) node[below,midway]{\tiny{\textcolor{black}{$|i^\ma(x_1)|$}}};

\draw[-, color=black] (2.5, 0.1)--(2.5,-0.1) node[below]{\textcolor{black}{\scriptsize{$x_1'$}}};


\draw[->, color=black, dashed] (2.5,0).. controls (2.75,0.5) .. (3,0);

\draw[-, color=black] (8.5, 0.1)--(8.5,-0.1) node[below]{\textcolor{black}{\scriptsize{$x_2$}}};
\draw[<-, color=black] (7.5,0).. controls (7.75,0.5) and (8.25,0.5).. (8.5,0);
\draw [color=black,
	decoration={
        brace,
        mirror,
        raise=0cm
    },
    decorate
] (7.5,-0.4) -- (8.5,-0.4) node[below,midway]{\tiny{\textcolor{black}{$|i^\mi(x_2)|$}}};

\draw[-, color=black] (6.5, 0.1)--(6.5,-0.1) node[below]{\textcolor{black}{\scriptsize{$x_2'$}}};

\draw[->, color=black, dashed] (6.5,0).. controls (6.25,0.5) .. (6,0);
\end{scope}
\end{tikzpicture}

\caption{Structure of the optimal policy}\label{fig:strategy}
\end{figure}
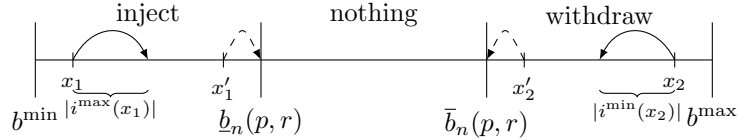

\subsection{Structure considering some special cases}\label{ssec:specialcases}
After considering the structure of the optimal policy in general and realizing that it is not of bang-bang type, we focus now on a special case, where we will see a bang-bang structure. Assume that $b^\mi-x \le i_n^\mi(x) \le i_n^\ma(x) \le b^\ma-x$ for each $n  \in \{0,1,\ldots ,N\}$ and $x\in[b^\mi, b^\ma]$. Thus the restriction set simplifies to
\begin{align}
D_n(x) = \{a\in A| i_n^\mi(x) \le a \le i_n^\ma(x)\}.
\end{align}
We further assume that $i_n^\mi$ and $i_n^\ma$ are linearly decreasing functions. Notice that this special case includes the case $ i_n^\mi(x)=b^\mi-x$ and $i_n^\ma(x)= b^\ma-x$, where it is possible to empty and fill the storage within one time step, sometimes called the ``fast storage''.

Additionally the terminal reward function is supposed to be a linear function of $x$.

\begin{theorem}
In the previously described special case, the value function $V_n$ is a linear function in $x$, that is
\begin{align}
V_n(x,p,r) = c_n(p,r) x+ d_n(p,r) \text{ for each } n \in \{0,1,\ldots ,N\}, p \in \cP, r \in S_R
\end{align}
where $c_n$ and $d_n$ are suitable functions. Further an optimal policy is of the form
\begin{align}
f^\star_n(x,p,r) = \begin{cases} i_n^\mi(x), & \gamma^1_n(p,r) < 0, \\
0, & \gamma^2_n(p,r) \le 0 \le \gamma^1_n(p,r),\\
i_n^\ma(x) , & \gamma^2_n(p,r) >0.\end{cases}
\end{align}
with suitable functions $\gamma^2_n(p,r) \le \gamma^1_n(p,r)$.
\end{theorem}

\begin{proof}
First notice that for functions $c_n$ and $d_n$ there exist functions $g_{n-1}$ and $l_{n-1}$ such that
\begin{align*}
\E_{n-1,p,r}[c_n(P_n, R_n)]= g_{n-1}(p,r)
\intertext{ and }
\E_{n-1,p,r}[d_n(P_n, R_n)]= l_{n-1}(p,r).
\end{align*}
We use backward induction. For $t=N$ we get by the linearity of $h_N$ in $x$
$$V_N(x,p,r) = h_N(x,p) = c_N(p,r) x+ d_N(p,r).$$
Assuming that $V_n(x,p,r) = c_n(p,r) x+ d_n(p,r)$ and
going backwards in time from $n$ to $n-1$ we get
\begin{align*}
V_{n-1}(x,p,r) &= \max && \hspace{-12mm} \left(\sup\limits_{i_{n-1}^\mi (x) \le a \le 0} \{-e(p) a + \E [V_n(x+a,P_n, R_n)|P_{n-1}=p, R_{n-1}=r]\}, \right.\\
& && \hspace{-9mm} \left. \sup\limits_{0\le a \le i_{n-1}^\ma(x)} \{-k(p) a + \E [V_n(x+a,P_n, R_n)| P_{n-1}=p, R_{n-1}=r]\}\right)\\
&= \max &&\hspace{-12mm}  \left(\sup\limits_{i_{n-1}^\mi (x) \le a \le 0} \{-e(p) a +g_{n-1}(p,r)(x+a) +  l_{n-1}(p,r)\}\,, \right. \\
& && \hspace{-9mm} \sup\limits_{0\le a \le i_{n-1}^\ma(x)} \{-k(p) a + g_{n-1}(p,r)(x+a) +  l_{n-1}(p,r)\}\Bigg)\\
 &= xg_{n-1}(p,r) &&\hspace{-4mm} + l_{n-1}(p,r) + \max  \left(  (g_{n-1}(p,r) -e(p)) i_{n-1}^\mi(x) \ind_{\{g_{n-1}(p,r)-e(p)< 0\}}, \right. \\
 & && \hspace{28mm} \left.  (g_{n-1}(p,r)-k(p)) i_{n-1}^\ma(x) \ind_{\{g_{n-1}(p,r)-k(p) >0\}}\right)\\
 & = x g_{n-1}(p,r) &&\hspace{-4mm} + l_{n-1}(p,r) + (g_{n-1}(p,r) -e(p)) i_{n-1}^\mi(x) \ind_{\{g_{n-1}(p,r)-e(p) <0\}} \\
 &&&\hspace{-4mm} + (g_{n-1}(p,r) -k(p)) i_{n-1}^\ma(x) \ind_{\{g_{n-1}(p,r)-k(p) >0\}}\\
 &=:xc_{n-1}(p,r)&& \hspace{-4mm} + d_{n-1}(p,r).
\end{align*}
The last equality holds since $i_n^\mi$ and $i_n^\ma$ are linear functions of $x$. The last but one equality holds since the two indicator functions are never equal to $1$ at the same time. This is true since in that case we would have $k(p) < g_{n-1}(p,r) < e(p)$ which is a contradiction to $k(p) \ge e(p)$.
Further we see that
$$f^\star_{n-1}(x,p,r) = \begin{cases}
i_{n-1}^\mi(x), & \gamma^1_{n-1}(p,r) < 0 ,\\
0, & \gamma^2_{n-1}(p,r) \le 0 \le \gamma^1_{n-1}(p,r)\\
i_{n-1}^\ma(x), & \gamma^2_{n-1}(p,r)>0
\end{cases}$$
where $\gamma^1_{n-1}(p,r) = g_{n-1}(p,r) -e(p)$ and $\gamma^2_{n-1}(p,r) = g_{n-1}(p,r) -k(p)$.\\
Since $e(p) \le k(p)$ we have
\begin{align*}
\gamma^1_{n-1}(p,r) = g_{n-1}(p,r) -e(p) \ge g_{n-1}(p,r) -k(p) = \gamma^2_{n-1}(p,r).
\end{align*}

\end{proof}

\begin{remark}
From the proof we can see that in the case $k(p) = e(p)$ we get $\gamma^1(p,r) = \gamma^2(p,r)$. Thus ``doing nothing'' is optimal only if $\gamma^1(p,r)= \gamma^2(p,r)=0$. Also note that for monotonicity properties of $\gamma^1_{n}(p,r), \gamma^2_{n}(p,r)$ in the price $p$, further assumptions on the price process (like in \cite{Sec10}) are necessary.
\end{remark}

\section{Algorithms}\label{sec:algorithms}
In the following we set the discount factor $\beta=1$ for simplification.

\subsection{General structure of the algorithms}
The general structure of any algorithm to compute the value of the gas storage facility is given by a Backward Induction Algorithm that is based on the Bellman equation (\ref{eq:bellman}) and Theorem \ref{thm:optimalpolicy}.\\

Given a grid in the gas volume level of the storage $x$ and the gas price $p$ we first compute the value function at time $N$ for all grid points in $x$ and $(p,r)$. Going backward in time we then compute the policy bounds $\ul b$ and $\ol b$ by solving the corresponding optimization problems for all $(p,r)$ and finally compute the maximizer and value function at that stage for all $x$ and $(p,r)$. Figure \ref{fig:generalstructure} illustrates the general structure. The arrows indicate \texttt{for}-loops.\\

\begin{figure}[htbp]
\centering
\begin{tikzpicture}[xscale=1.2, yscale=1.2]

\shadedraw[top color=black!12, bottom color= black!12, draw= black] (0.3,-0.3) rectangle +(9.4,-1.4);
\draw[->, color= black, thick] (9.8,-1.6) -- (10.7,-1.6)-- (10.7, -0.4)node[midway, right] {\small{$\forall \ (p,r)$}}  -- (9.8,-0.4) ;

\shadedraw[top color= black!19, bottom color=  black!19, draw= black] (0.6,-0.6) rectangle +(8.8,-0.8);

\draw[->, color= black, thick] (9.5,-1.3) -- (10.1,-1.3)-- (10.1, -0.7)node[midway, right] {\small{$\forall \ x$}} -- (9.5,-0.7);

\node[right] at (0.7, -1) {\small{Compute $V_N(x,p,r) = h_N(x,p)$}};

\draw[->, thick] (3.6, -1.7) -- (3.6, -1.95);

\shadedraw[top color= black!5, bottom color=  black!5, draw= black] (0,-2) rectangle +(10,-6.5);

\draw[->, color= black, thick] (10.1,-8.4) -- (11.8,-8.4)-- (11.8, -2.1)node[rotate=270] at(12.1,-5.25){\small{$\text{for } n=(N-1), \ldots,  0$}}  -- (10.1,-2.1) ;

\shadedraw[top color= black!12, bottom color=  black!12, draw= black] (0.3,-2.3) rectangle +(9.4,-5.9);

\draw[->, color= black, thick] (9.8,-8.1) -- (10.7,-8.1)-- (10.7, -2.4)node[midway, right] {\small{$\forall \ (p,r)$}}  -- (9.8,-2.4) ;

\node[right] at (0.4, -2.7) {\small{Compute $\ul b_n(p,r), \ \ol b_n(p,r)$ by:}};
\node at (3.6, -3.4) {\small{$\max\limits_{x\in [b^\mi, b^\ma]} U_n(x,p,r) - k(p) x$} };
\node at (3.6, -4.1) {\small{$\max\limits_{x\in [b^\mi, b^\ma]} U_n(x,p,r) - e(p) x$} };

\shadedraw[top color= black!19, bottom color= black!19, draw= black] (0.6,-4.6) rectangle +(8.8,-3.3);

\draw[->, color=black, thick] (9.5,-7.8) -- (10.1,-7.8)-- (10.1, -4.7)node[midway, right] {\small{$\forall \ x$}} -- (9.5,-4.7);

\node[right] at (0.7, -4.9) {\small{Compute}};

\node[right] at (0.8, -5.9){\small{$ f^\star_n(x,p,r) = \begin{cases}
\min(\ul b_n(p,r) -x, i_n^\ma(x)), & b^\mi \le x < \ul b_n(p,r),\\
0 , & \ul b_n(p,r) \le x \le \ol b_n(p,r),\\
\max(\ol b_n(p,r) -x, i_n^\mi(x)), & \ol b_n(p,r) < x \le b^\ma.
\end{cases} $}};

\node[right] at (0.7, -6.9) {\small{and}};

\node[right] at (0.8, -7.5) {\small{$V_n(x,p,r) = h(p, f_n^\star) + U_n(x+ f_n^\star, p,r)$}};

\end{tikzpicture}

\caption{General Structure of any algorithm using the optimal policy}\label{fig:generalstructure}
\end{figure}
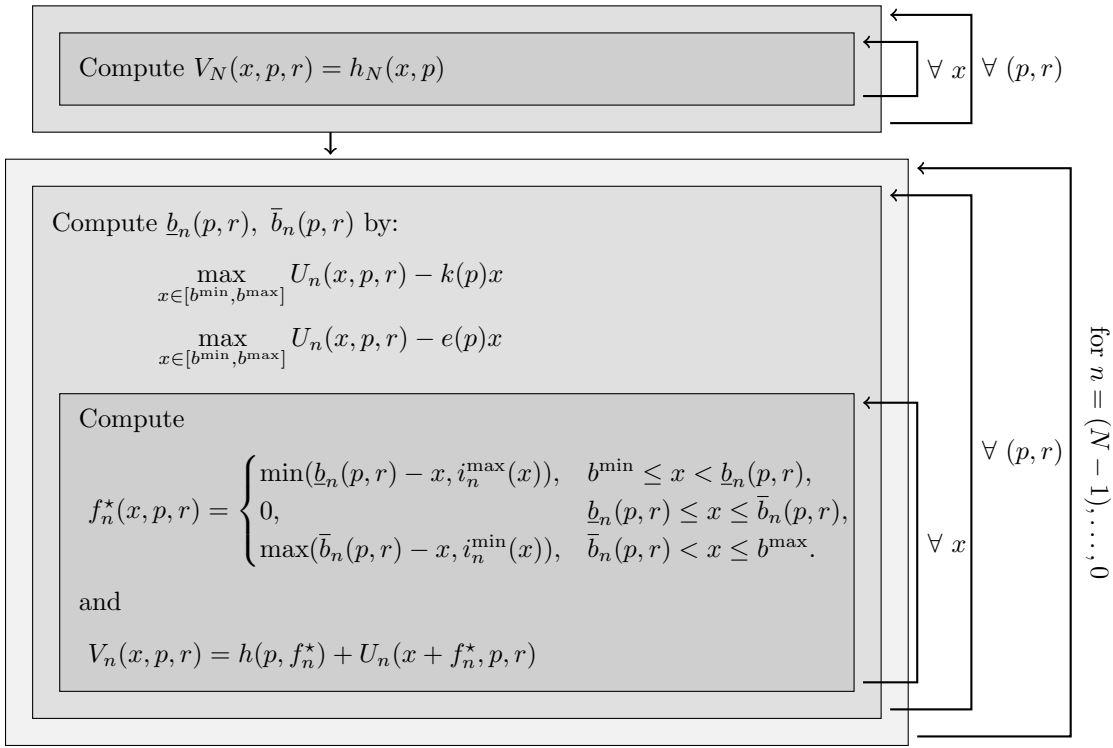

Looking at the basic structure of the algorithm it becomes clear, that the key question is how to compute or estimate $U_n(x,p,r)= \E_{n,p,r}[V_{n+1}(x, P_{n+1}, R_{n+1})]$. In the end every other parameter is given or can be calculated based on $U_n$.

\begin{remark}[Potential of using the optimal policy]
Comparing the basic structure of the algorithm illustrated in figure \ref{fig:generalstructure} to the structure directly using the Bellman Equation \ref{eq:bellman} we see that instead of solving optimization problems for all $(x,p,r)$ it is enough to solve two optimization problems for all $(p,r)$. Thus we get a significant reduction of the number of optimization problems which have to be solved.
This and the fact that the optimization problems involve an optimization over the whole interval $[b^\mi, b^\ma]$ instead of $D_n$ which depends on the current volume level yields a good potential for parallelization in implementing any algorithm using the structure of the optimal policy.
\end{remark}

\subsection{Specification of the price process}
In order to show possible ways to compute or estimate $U_n$ and thus develop algorithms we introduce a specific model of the price process. Later some remarks will indicate to which extend the model can be generalized such that the presented algorithm still works. We will consider a continuous-time model which we will discretize to use our MDP approach.\\

We assume that the log-prices follow a regime-switching variation of the one-factor Schwartz' model, see \cite{Schwartz97}, with a possible switch in the time varying mean. It has been shown by \cite{chen2010implications} that regime-switching models fit real data very well. In what follows suppose that $R_t := R_{\lfloor t\rfloor}$ is the continuous process constructed from the discrete-time Markov chain $(R_n)$. Thus, the log-price-process $(X_t)_{t\ge 0}$ satisfies the stochastic differential equation
\begin{align*}
dX_t = \alpha(\mu_{R_t}(t)-X_t) dt + \sigma dB_t,
\end{align*}
where $\alpha \in \Rp$ is the mean-reversion factor, $\mu_r(t)$ is the time dependent mean when the Markov Chain $(R_t)$ is in state $r$, $\sigma \in \Rp$ denotes the volatility and $(B_t)_{t\ge 0}$ is a Brownian Motion.
The price process $(P_t)_{t\ge 0}$ itself is given by $P_t=e^{X_t}$. In the description of the two algorithms below, we assume that the state space of $(R_t)$ consists of two different values only.

\begin{remark}
Although we assume that the price process can only switch between two regimes, the algorithms below can easily be extended to the case of more regimes, as well as faster switches.
\end{remark}

\begin{example}\label{ex:parameter-price}[Choice of the parameters]
The following parameters were chosen for the numerical examples. Based on a parameter estimation by \cite{BBK08}, who fitted a similar model without regime switching to NBP data. We assume that
\begin{align*}
\mu_1(t) & = a_0 + a_1 t + a_2 \cos(2\pi(t-a_3)/250), \\
\mu_2(t) &= a_0 - a_1 t + a_2 \cos(2\pi(t-a_3)/250), \\
\end{align*}
where $a_0 = 2.69$, $a_1=0.0007$, $a_2 = -0.234$ and $a_3= 118.1$. The mean reversion factor $\alpha$ is assumed to be  $0.073$ and the volatility to be $\sigma=0.072$. Further the transition matrix of the Markov chain is given by
$$Q= \begin{pmatrix}
0.9 & 0.1\\
0.5 &0.5
\end{pmatrix}.$$
\end{example}

Having specified the price model we now introduce two methods to get a grid in the price domain and compute $U_n$.

\subsection{A ``Multinomial-Tree'' Algorithm}
The ``Multinomial-Tree'' method is based on the idea of approximating the price process by a recombining Binomial tree and thus calculate the conditional expectation $U_n$ ``exactly'' on that grid. The authors of \cite{Nel90} developed a method for approximating diffusion models, which we will use to get an approximating recombining tree for the log-price process. The advantage of a recombining tree is that the number of nodes grows linearly and not exponentially. For this purpose the time interval $[0,N]$ is divided into equidistant parts of length $\Delta t$. Let for each $y\in \R$, $j = 0,1,\ldots, \frac N \dt$, $r=1,2$
\begin{align*}
&{}^\up\qdt^r(y, j) \text{ where } 0 \le {}^\up\qdt^r(y, j)  \le 1,\\
&{}^\down\Ydt^r(y, j), {}^\up\Ydt^r( y, j ) \text{ where } - \infty < {}^\down\Ydt^r(y, j)\le {}^\up\Ydt^r( y, j )< \infty,
\end{align*}
be functions on $\R\times\N_0$ that indicate the probability of the ``Up''-movement at gridpoint $\dt\cdot j$ starting at $y$, as well as the values that are reach by an ``Up''- and a``Down''-movement respectively
The process $(\Ydt^r(t))_{t\ge 0}, \ r=1,2$ is given by
\begin{align*}
&\Ydt^r (0) = X_0, \\
&\Ydt^r (t) = \Ydt^r(j ), \quad  j\dt \le t <(j+1) \dt,\\
&\P[\Ydt^r(j+1 ) = {}^\up\Ydt^r(\Ydt^r(j), j)| j\dt,\Ydt^r(j) ] = {}^\up\qdt^r(\Ydt^r(j), j),\\
&\P[\Ydt^r(j+1 ) = {}^\down\Ydt^r(\Ydt^r(j), j)| j\dt,\Ydt^r(j) ] = {}^\down\qdt^r(\Ydt^r(j), j).
\end{align*}
This process, or the corresponding functions, are now constructed such that we get a recombining structure and the local drift and local second moment converge to the drift and volatility of the price process as $\dt$ tends to $0$.\\

Using the rather intuitive choice of ${}^{\up/\down}\Ydt^r(y,j \dt) = y \pm \sqrt{\dt} \sigma$ yields a recombining structure and by \cite{Nel90} the ``underlying'' price tree. Note that ${}^{\up/\down}\Ydt^r$ only depends on $\sigma$ which is not influenced by a regime switch, thus the grid in the price domain for the algorithms is the same in each regime. Figure \ref{fig:pricetree} illustrates the underlying structure of the Binomial tree process.

\begin{figure}[htbp]
\centering
\begin{tikzpicture}[xscale=1, yscale=1]
\node[left] (h) at (0,0) {\small{$y$}};
\node (hu) at (2,1) {\small{$y+ \sigma\sqrt{\dt}$}};
\node (hd) at (2,-1) {\small{$y- \sigma\sqrt{\dt}$}};
\node[right](hdd)  at (4,-2) {\small{$y- 2\sigma\sqrt{\dt} $}};
\node[right] (hud) at (4,0) {\small{$y$}};
\node[right](huu)  at (4,2) {\small{$y+2\sigma \sqrt{\dt}$}};

\node[right](hn)  at (8,4) {\small{$y+n\sigma \sqrt{\dt}$}};
\node[right](hn-2)  at (8,3) {\small{$y+(n-2)\sigma \sqrt{\dt}$}};

\node[right](h-n)  at (8,-4) {\small{$y-n\sigma \sqrt{\dt}$}};
\node[right](h-n-2)  at (8,-3) {\small{$y-(n-2) \sigma\sqrt{\dt}$}};

\draw[->] (h)--(hu) ;
\draw[->]  (h)--(hd) ;
\draw [->] (hu)--(huu);
\draw [->] (hu)--(hud);
\draw [->] (hd)--(hdd);
\draw [->] (hd)--(hud);

\draw [-, dotted] (8.2,2.8)--(8.2,-2.8);

\draw [-, dotted] (huu)--(7,3.5) -- (8,4);
\draw [-, dotted] (7,3.5) -- (8,3);

\draw [-, dotted] (hdd)--(7,-3.5) -- (8,-4);
\draw [-, dotted] (7,-3.5) -- (8,-3);

\draw [-, dotted] (hud)--(5,0.5);
\draw [-, dotted] (hud)--(5,-0.5);

\end{tikzpicture}
\caption{Structure of the ``underlying'' price tree}\label{fig:pricetree}
\end{figure}
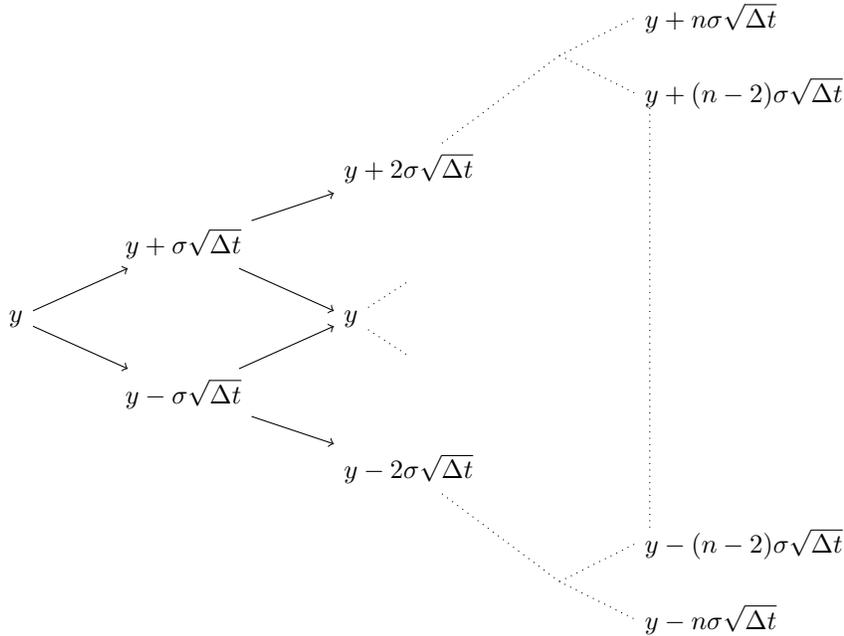

Based on \cite{Nel90} the probabilities are chosen to be

\begin{align*}
{}^\up\qdt^r(y, j)& = \max\left( 0, \min \left(1, \widetilde{{}^\up\qdt^r}(y, j) \right)\right), \quad {}^\down\qdt^r(y, j)= 1- {}^\up\qdt^r(y, j).
\end{align*}
where

\begin{align*}
\widetilde{{}^\up\qdt^r}(y, j)&  = \frac{\dt \cdot \alpha(\mu_r(t) - y) + y -{}^\down\Ydt(y, j) }{{}^\up\Ydt(y, j) - {}^\down\Ydt(y, j) }. \\
\end{align*}

Under some further assumption it has been shown in \cite{Nel90} that the process $(\Ydt^r)_{t\ge 0}$ converges weakly to to the corresponding log-price process as $\dt$ tends to $0$. Thus its exponential converges weakly to the price process.\\

In order to use this construction in an algorithm to compute
$$U_n(x,p,r) = \sum\limits_{l \in S_R} q_{rl} \E_{n,p,r}[V_{n+1}(x, P_{n+1}, l)]$$
we pick some fixed value $m\in \N$ and choose $\dt$ such that $m \cdot \dt= 1$.
Approximating the price process with the method from above we know that the random variable $\widetilde P_{n+1}$ given $\widetilde P_n=p$ can take $m+1$ different values denoted by $p_{n+1}^{(p,1)}, \ldots,p_{n+1}^{(p,m+1)}, $ with positive probability. Figure \ref{fig:pircetreeinalgorithm} illustrates this for $m=1, 2,3$.
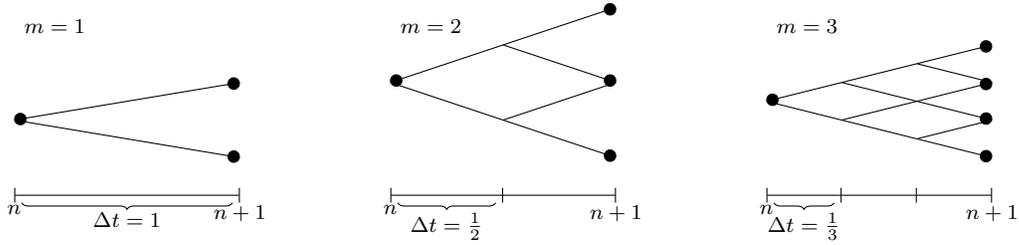
\begin{figure}[htbp]
\centering
\begin{tikzpicture}[xscale=1, yscale=1]

\draw[|-|] (0,0) node[below]{\scriptsize{$n$}}--(3,0)node[below]{\scriptsize{$n+1$}} ;

\draw[|-|] (5,0)node[below]{\scriptsize{$n$}}--(6.5,0);
\draw[-|] (6.5,0)--(8,0)node[below]{\scriptsize{$n+1$}} ;

\draw[|-|] (10,0)node[below]{\scriptsize{$n$}}--(11,0);
\draw[-|] (11,0)--(12,0);
\draw[-|] (12,0) --(13,0)node[below]{\scriptsize{$n+1$}} ;

\draw [color=black,
	decoration={
        brace,
        mirror,
        raise=0cm
    },
    decorate
] (0.1,-0.1) -- (2.9,-0.1) node[below,midway]{\tiny{\textcolor{black}{\scriptsize{$\dt= 1$}}}};

\draw [color=black,
	decoration={
        brace,
        mirror,
        raise=0cm
    },
    decorate
] (5.1,-0.1) -- (6.4,-0.1) node[below,midway]{\tiny{\textcolor{black}{\scriptsize{$\dt=\frac 12$}}}};

\draw [color=black,
	decoration={
        brace,
        mirror,
        raise=0cm
    },
    decorate
] (10.1,-0.1) -- (10.9,-0.1) node[below,midway]{\tiny{\textcolor{black}{\scriptsize{$\dt=\frac 13$}}}};

\draw[*-*] (0,1) -- (3,1.5);
\draw[-*] (0,1) -- (3,0.5);

\draw[*-] (5,1.5) -- (6.5,2);
\draw[-] (5,1.5) -- (6.5,1);
\draw[-*] (6.5,2) -- (8,2.5);
\draw[-*] (6.5,2) -- (8,1.5);
\draw[-*] (6.5,1) -- (8,0.5);
\draw[-] (6.5,1) -- (8,1.5);

\draw[*-] (10,1.25) -- (11,1.5);
\draw[-] (10,1.25) -- (11,1);
\draw[-] (11,1.5) -- (12,1.75);
\draw[-] (11,1.5) -- (12,1.25);
\draw[-] (11,1) -- (12,1.25);
\draw[-] (11,1) -- (12,0.75);
\draw[-*] (12,1.25) -- (13,1.5);
\draw[-*] (12,1.25) -- (13,1);
\draw[-*] (12,0.75) -- (13,0.5);
\draw[-] (12,0.75) -- (13,1);
\draw[-] (12,1.75) -- (13,1.5);
\draw[-*] (12,1.75) -- (13,2);

\node[right] at (0,2.25) {\scriptsize{$m=1$}};

\node[right] at (5,2.25) {\scriptsize{$m=2$}};
\node[right] at (10,2.25) {\scriptsize{$m=3$}};

\end{tikzpicture}
\caption{Structure of the price tree using different values of $\dt$}\label{fig:pircetreeinalgorithm}
\end{figure}

Thus the continuation value $U_n$ on that tree can be calculated by
\begin{align}\label{eq:U-binalgo}
U_n(x,p,r) = \sum\limits_{l\in S_R} q_{rl} \sum\limits_{j = 1}^{m+1} V_{n+1}(x, p_{n+1}^{(p,j)},l) \cdot \P(\widetilde P_{n+1}= p_{n+1}^{(p,j)}| \widetilde P_n=p, R_n=r) .
\end{align}

The necessary probabilities are computed using ${}^{\up/\down}\qdt^r$ from the contruction. For example in the case $m=1$ we simply have $\P(\widetilde P_{n+1}= p_{n+1}^{(p,1)}| \widetilde P_n=p, R_n=r) = {}^\up \qdt^r(\log(p),n)$ and $\P(\widetilde P_{n+1}= p_{n+1}^{(p,2)}| \widetilde P_n=p, R_n=r) = {}^\down \qdt^r(\log(p),n)$. In the case $m=2$ we get
\begin{align*}
\P(\widetilde P_{n+1}= p_{n+1}^{(p,1)}| \widetilde P_n=p, R_n=r) &= {}^\up \qdt^r(\log(p),n) \cdot {}^\up \qdt^r\left(\log(p) + \sigma \sqrt{\frac 12},n+\frac 12\right) ,\\
\P(\widetilde P_{n+1}= p_{n+1}^{(p,2)}| \widetilde P_n=p, R_n=r)&= {}^\up \qdt^r(\log(p),n) \cdot {}^\down \qdt^r\left(\log(p) + \sigma \sqrt{\frac 12},n+\frac 12\right) \\
& + {}^\down \qdt^r(\log(p),n) \cdot {}^\up \qdt^r\left(\log(p) - \sigma \sqrt{\frac 12},n+\frac 12\right),\\
\P(\widetilde P_{n+1}= p_{n+1}^{(p,3)}| \widetilde P_n=p, R_n=r) &= {}^\down \qdt^r(\log(p),n) \cdot {}^\down \qdt^r\left(\log(p) - \sigma \sqrt{\frac 12},n+\frac 12\right) .
\end{align*}

Because of the recombining structure we get at time $n$ at most $1+m\cdot n$ gridpoints in $p$. There are \textit{at most} such many gridpoints since some are reached with probability $0$ and thus do not need to be considered in the calculation.

\begin{remark}[More general price process models]
The construction of a recombining tree works according to \cite{Nel90} for any diffusion model of the form
$$d X_t= \mu(X_t,t) dt + \sigma(X_t,t) dB_t,$$ although the construction of the ``underlying'' price tree is more complicated in some cases.  To add a possible regime switching it is necessary that the switch in regimes does not affect the underlying price tree, but only the probabilities. That means we can allow for a switch in mean, not in the volatility. But as we see, it is possible to allow for a time dependent volatility, for example.
\end{remark}

\begin{remark}[Convergence of Multinomial-Tree Algorithm]
It can be shown that with $m\to \infty$, i.e. the tree is getting finer, the algorithm converges against the true value and the optimal policy converges against the true optimal one. This is essentially due to the fact that according to \cite{Nel90} the tree converges weakly against the true diffusion price process. However, a thorough proof needs a number of technical steps which we will indicate in this remark. In order to do so, let us denote by $V_n^{(m)}$ the value function for the tree price process with $m$ nodes per time step and by $V_n$ the value function for the model with diffusion price process. By $f_n^{(m)}$ and $f_n$ we denote the corresponding optimal decision rules.

First note that the set $D_n(x)$ of admissible actions does not depend on the tree, is compact and the set-valued mapping $x\mapsto D_n(x)$ is continuous in the sense of \cite{BR11} (Definition A.2.1). Thus it follows with Theorem 2.4.10 in \cite{BR11} that $V_n^{(m)}(x,p,r)$ as well as $V_n(x,p,r)$ are continuous in $x$ and $p$ for all $m$. Next it is possible to find for all $n$ a r.v. $Y_n$ which is larger w.r.t. the increasing convex order than the price $P_n^{(m)}$ for all $m$. Finally for all $m,x$ and $r$ we get that $V_n^{(m)}(x,p,r) \le p d_1+d_2$ for positive constants $d_1, d_2$, since the value of the gas storage can be bounded above by a system where we can sell the complete storage at once at every stage.

Then the statement can be shown by induction using Theorem A.1.5 in \cite{BR11}: Since $V_N^{(m)} = h_N$ independent of $m$, we have the stated convergence for $N$. Now suppose $V_{n+1}^{(m)}\to V_{n+1}$ for $m\to\infty$ and the maximizer at stage $n+1$ also converges (in the sense that $Ls D_{n+1}^{*(m)}(x) \subset D_{n+1}^*(x)$, see  Theorem A.1.5 in \cite{BR11}). Then it can be shown that $$ \sum_{j}q_{rj}\int V_{n+1}^{(m)}(x+a,p',j)Q^{(m)}_{n,r}(dp'|p) \to \sum_{j}q_{rj}\int V_{n+1}(x+a,p',j)Q_{n,r}(dp'|p)$$ for $m\to\infty$ and $Ls D_{n}^{*(m)}(x) \subset D_{n}^*(x)$.

\end{remark}

\subsection{A Least-Square Monte Carlo Algorithm}

The second method we present is based on the well known Least-Square Monte Carlo Method of \cite{LS01} that was first used for gas storage valuation by \cite{BdJ08}. We take up that idea and adapt it to the problem with regime switching and moreover use the structure of the optimal policy.
Since the whole method is based on a Monte Carlo simulation of the price process, we assume that there are $M$ paths of the price process simulated. Thus we have $M$ paths of the underlying Markov chain, that indicate the regimes, and $M$ paths of the actual corresponding price.
The idea is based on the fact, that since $U_n(x,\cdot,r)$ is a conditional expectation it is a function of the current price $p$ for the current regime $r$ and for each $x$. Thus we approximate this function by a linear combination of a set of $m$ basis function $\varphi_1^r,\ldots, \varphi_m^r$, i.e. we assume
\begin{align*}
U_n(x,p,r) = \sum\limits_{i=1}^m \beta_i^r \varphi_i^r(p),
\end{align*}
where $\beta_i^r$, $i=1,\ldots, m$ are to be chosen optimal in some sense. For calculating the coefficients $\beta_i^r$ a Least-Square method is used. At time point $n$ we have simulated $M$ prices $p_n^1,\ldots, p_n^M$ and $M$ corresponding regimes $r_n^1, \ldots, r_n^M$. Let $k=1, \ldots,\widetilde M$ be the indices where $r_n^k=1$ and $k=\widetilde M +1, \ldots,M$ those where $r_n^k=2$. For those two sets of indices we now use a Least-Square method to seperately compute $\beta_i^1$ and $\beta_i^2$ respectively. As realizations of the conditional expectation on the grid $p_n^1,\ldots, p_n^M$ we use the already calculated values $V_{n+1}(x,p_{n+1}^1, r_{n+1}^1), \ldots , V_{n+1}(x,p_{n+1}^M, r_{n+1}^M)$. Thus at time step $n$ and volume level $x$ we need to minimize
\begin{align}
\left\Vert
\begin{pmatrix}
V_{n+1}\big(x, p_{n+1}^{1}, r_{n+1}^{1}\big)\\
\vdots\\
V_{n+1}\big(x, p_{n+1}^{\widetilde M},r_{n+1}^{\widetilde M}\big)
\end{pmatrix}
-
\begin{pmatrix}
\vphi_1\big(p_n^{1}\big) &\cdots &\vphi_m \big(p_n^{1}\big)\\
\vdots & & \vdots \\
\vphi_1\big(p_n^{\widetilde M}\big) &\cdots &\vphi_m \big(p_n^{\widetilde M}\big)
\end{pmatrix}
\cdot
\begin{pmatrix}
\beta_1^1\\
\vdots\\
\beta_m^1
\end{pmatrix}
\right\Vert^2
\end{align}
to get $\beta_1^1,\ldots, \beta_m^1$ and
\begin{align}
\left\Vert
\begin{pmatrix}
V_{n+1}\big(x, p_{n+1}^{\widetilde M+1}, r_{n+1}^{\widetilde M+1}\big)\\
\vdots\\
V_{n+1}\big(x, p_{n+1}^{ M},r_{n+1}^{ M}\big)
\end{pmatrix}
-
\begin{pmatrix}
\vphi_1\big(p_n^{\widetilde M+1}\big) &\cdots &\vphi_m \big(p_n^{\widetilde M+1}\big)\\
\vdots & & \vdots \\
\vphi_1\big(p_n^{ M}\big) &\cdots &\vphi_m \big(p_n^{ M}\big)
\end{pmatrix}
\cdot
\begin{pmatrix}
\beta_1^2\\
\vdots\\
\beta_m^2
\end{pmatrix}
\right\Vert^2
\end{align}
to get $\beta_1^2,\ldots, \beta_m^2$, where $\Vert \cdot \Vert$ denotes the Euclidean norm.
In the algorithm we then estimate the conditional expectation by
\begin{align}
\hat U_n(x, p_n^j, r_n^j) = \sum\limits_{i=1}^m \beta_i^{r_n^j}\vphi(p_n^m).
\end{align}

\begin{remark}[Other price models]
This method can be used for any price model with regime switching, that can be simulated.
\end{remark}

\begin{remark}[Convergence of Least-Square Algorithm]
When the number of simulated paths $M$ and the number of chosen basis functions $m$ converges to infinity, then the Least-Square Algorithm converges to the true value. This result has been shown in \cite{clement2002analysis} for the general Least-Square Algorithm.
\end{remark}

\subsection{A Method of building a grid in the storage domain}
\subsubsection{Motivation.}
The most common and intuitive method of choosing a grid in the storage domain is an equidistant grid, i.e. dividing the interval $[b^\mi, b^\ma]$ into a fixed number of parts of the same length. But when we see the structure of the optimal policy we recognize that a lot of times the optimal policy is $a=i^\ma$ or $a=i^\mi$. That means that from state $x$ we often get to state $x+i^\ma(x)$ or $x+ i^\mi(x)$. These states are not necessarily grid points of an equidistant grid, especially when $i^\ma$ and $i^\mi$ are not constant. Thus we get errors in the computation, especially when the equidistant grid is not that fine.

Figure \ref{fig:bounds-equi} shows the policy bounds using the Mulitnomial-Tree Algorithm and the Least-Square Algorithm with an equidistant grid of $530$ grid points. The calculation is based on the example introduced in Section \ref{ssec:numeric}.  As we can see, the policy bounds equal $b^\ma$ or $b^\mi$ most of the time, which makes the decision to be $i^\ma$ or $i^\mi$. This emphasizes the idea of not choosing an equidistant grid, but to use the functions $i^\ma$ and $i^\mi$ and rather include points like $x+i^\ma(x)$ and $x+i^\mi(x)$ instead.

\begin{figure}[htbp]
\vspace*{0.3cm}
\centering
\includegraphics[trim = 15mm 10mm 15mm 10mm,angle={270}, width=0.45 \textwidth]{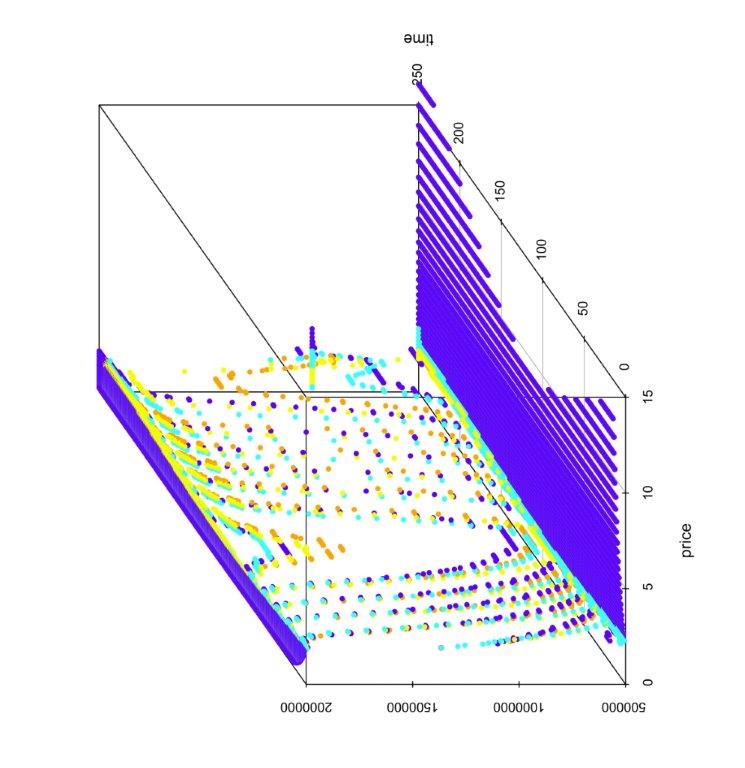}
\hspace{5mm}
\includegraphics[trim = 15mm 10mm 15mm 10mm,angle={270}, width=0.45 \textwidth]{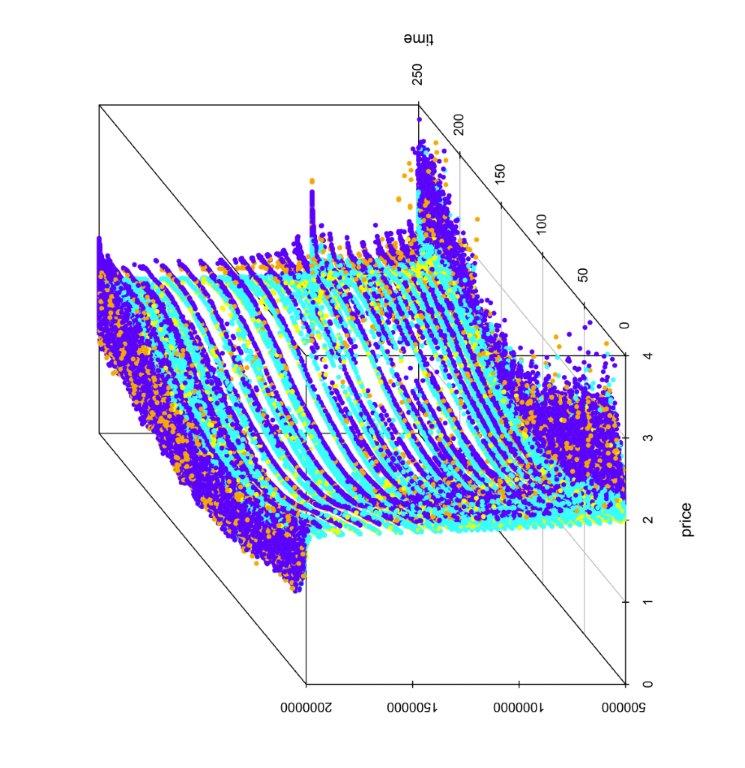}
\caption{Policy bounds $\protect\underline{b}_n(p,r)$ (lightblue, yellow) und $\protect\overline{b}_n(p,r)$ (darkblue, orange) using the Multinomial-Tree Algorithm (left) and the Least-Square Algorithm (right) with an equidistant grid}\label{fig:bounds-equi}
\end{figure}

\subsubsection{Construction and Example}
Clearly the grid we now construct should include the points $b^\ma, b^\mi, x_0$.
Starting from $b^\ma$ we compute $b^\ma + i^\mi(b^\ma), b^\ma + i^\mi(b^\ma)+ i^\mi(b^\ma + i^\mi(b^\ma)),\ldots $ until $b^\mi$ is reached. In the same way we compute starting from $b^\mi$  $b^\mi + i^\ma(b^\mi), b^\mi + i^\ma(b^\mi) + i^\ma(b^\mi + i^\ma(b^\mi)),\ldots$ until $b^\ma$ is reached. As far as $x_0$ does not equal $b^\ma$ nor $b^\mi$ we compute starting from $x_0$ the grid points $x_0 + i^\mi(x_0), x_0 + i^\mi(x_0)+ i^\mi(x_0 + i^\mi(x_0)), \ldots$ and $x_0 + i^\ma(x_0), x_0 + i^\ma(x_0)+ i^\ma(x_0 + i^\ma(x_0)), \ldots$ until the minimal and maximal volume level is reached respectively.

To get an impression of what the grid looks like, Figures \ref{fig:grid-points} and \ref{fig:grid-hist} show grids of minimal lengths $100, \ 500, \ 1500$ in the example we used in Chapter \ref{ssec:numeric}, that is with $b^\mi = 500 000 \text{MMBtu}, \  b^\ma = 2 000 000 \text{MMBtu}, x_0 = 1 000 000 \text{MMBtu}, i^\mi(x) = -70.71 \sqrt{x}, \ i^\ma(x) = -0.032 x+ 68170.$ \\

\begin{figure}[htbp]
\centering
\includegraphics[trim = 15mm 10mm 15mm 10mm,angle={270}, width= \textwidth]{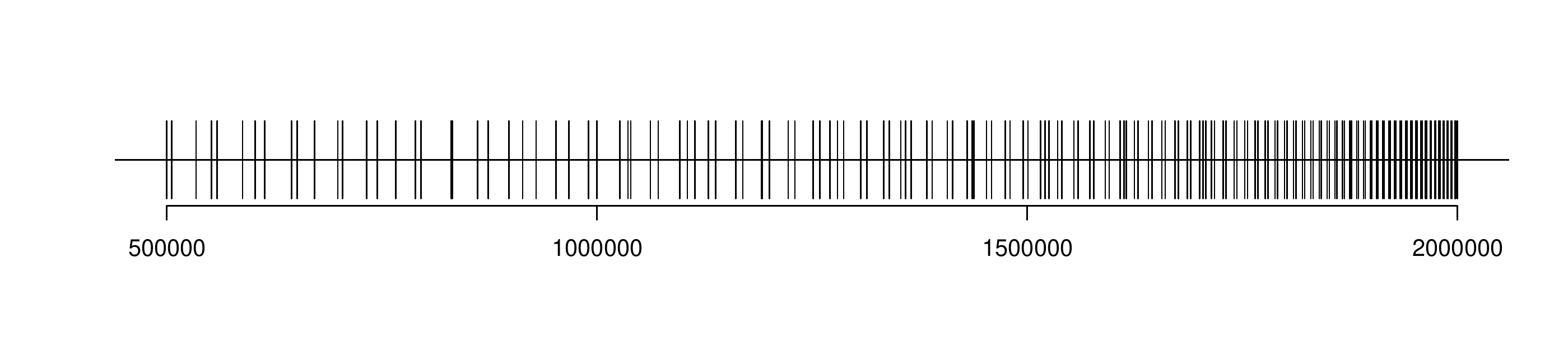}
\includegraphics[trim = 15mm 10mm 15mm 10mm,angle={270}, width= \textwidth]{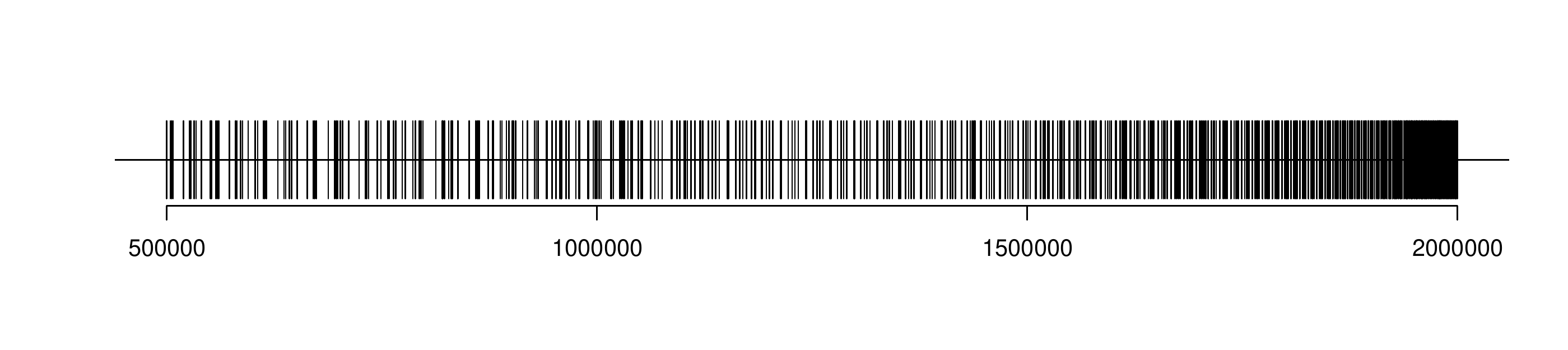}
\includegraphics[trim = 15mm 10mm 15mm 10mm,angle={270}, width= \textwidth]{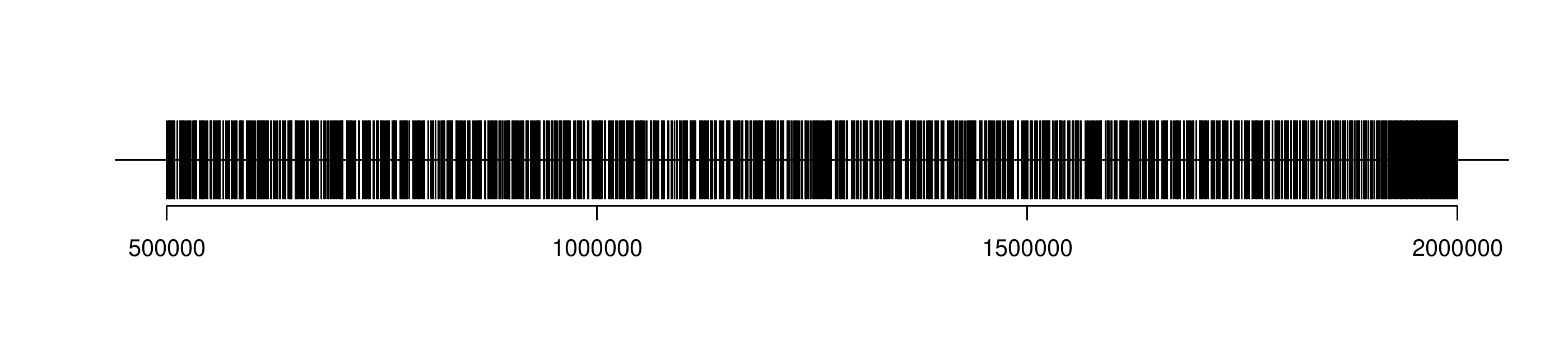}
\caption{Grid points of grids with minimal length $100,500, 1500$ and actual length $137, 530, 1506$}\label{fig:grid-points}
\end{figure}

\begin{figure}[htbp]
\centering
\includegraphics[trim = 15mm 10mm 15mm 10mm,angle={270}, width= 0.3\textwidth]{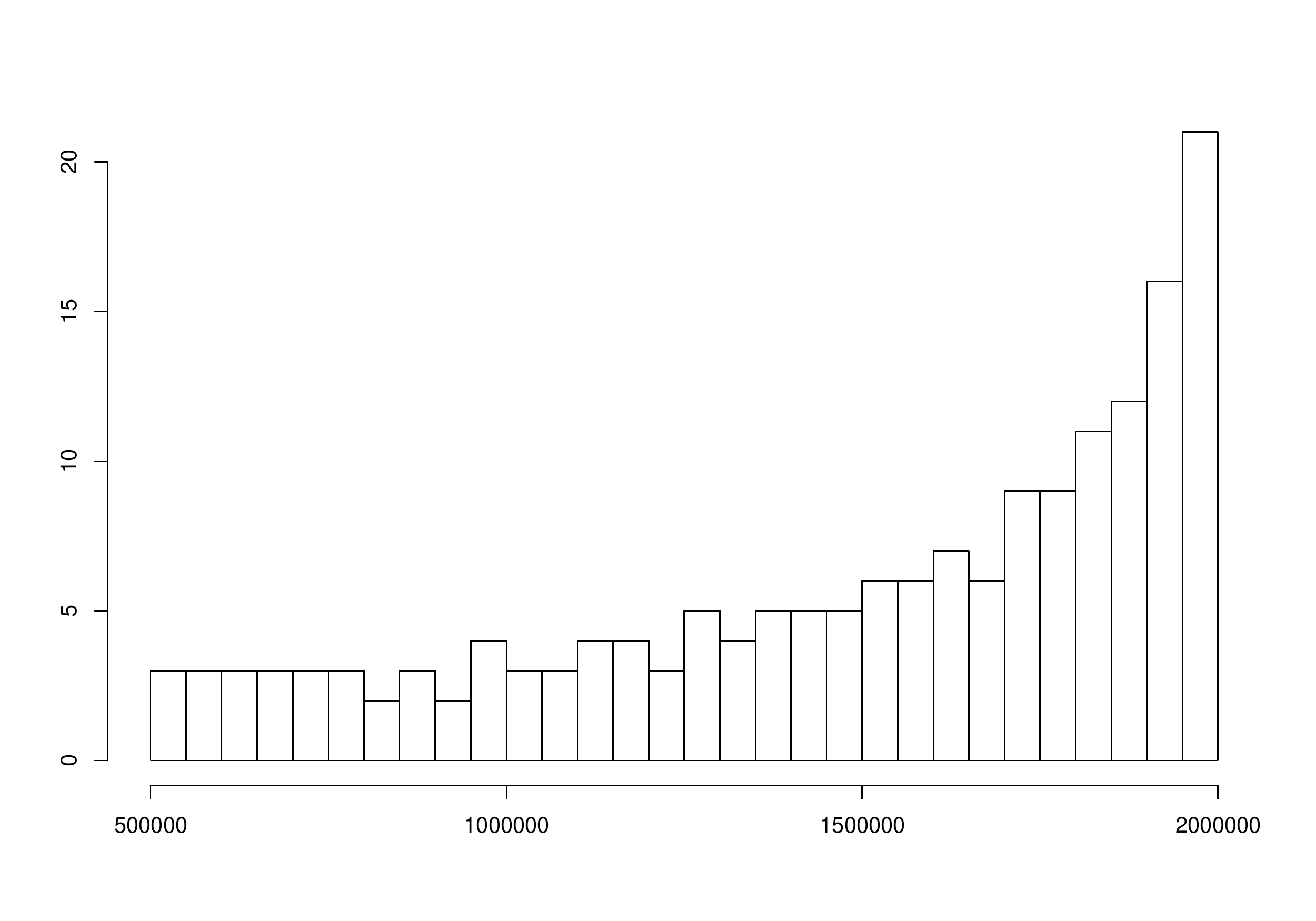}
\includegraphics[trim = 15mm 10mm 15mm 10mm,angle={270}, width=0.3 \textwidth]{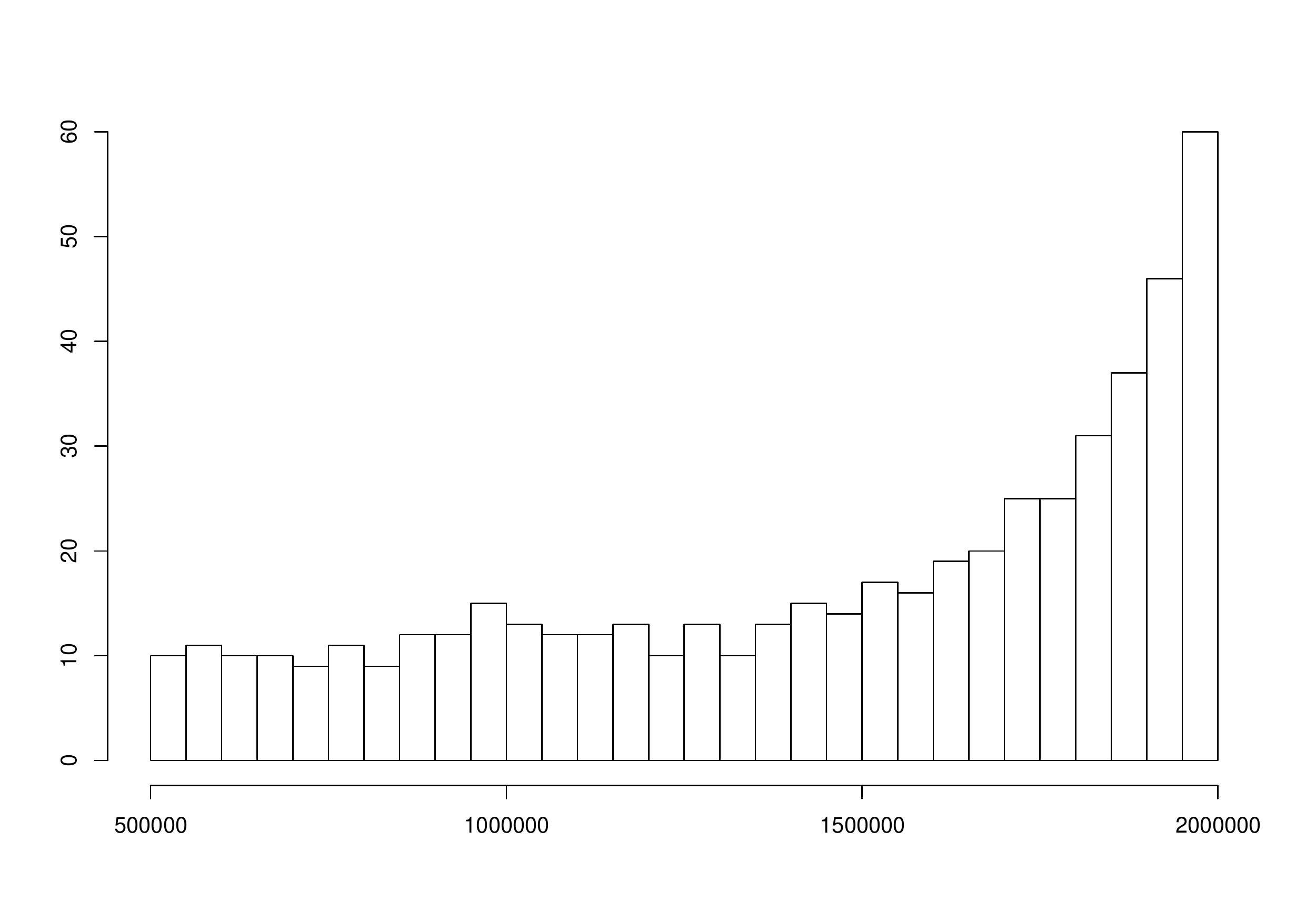}
\includegraphics[trim = 15mm 10mm 15mm 10mm,angle={270}, width= 0.3\textwidth]{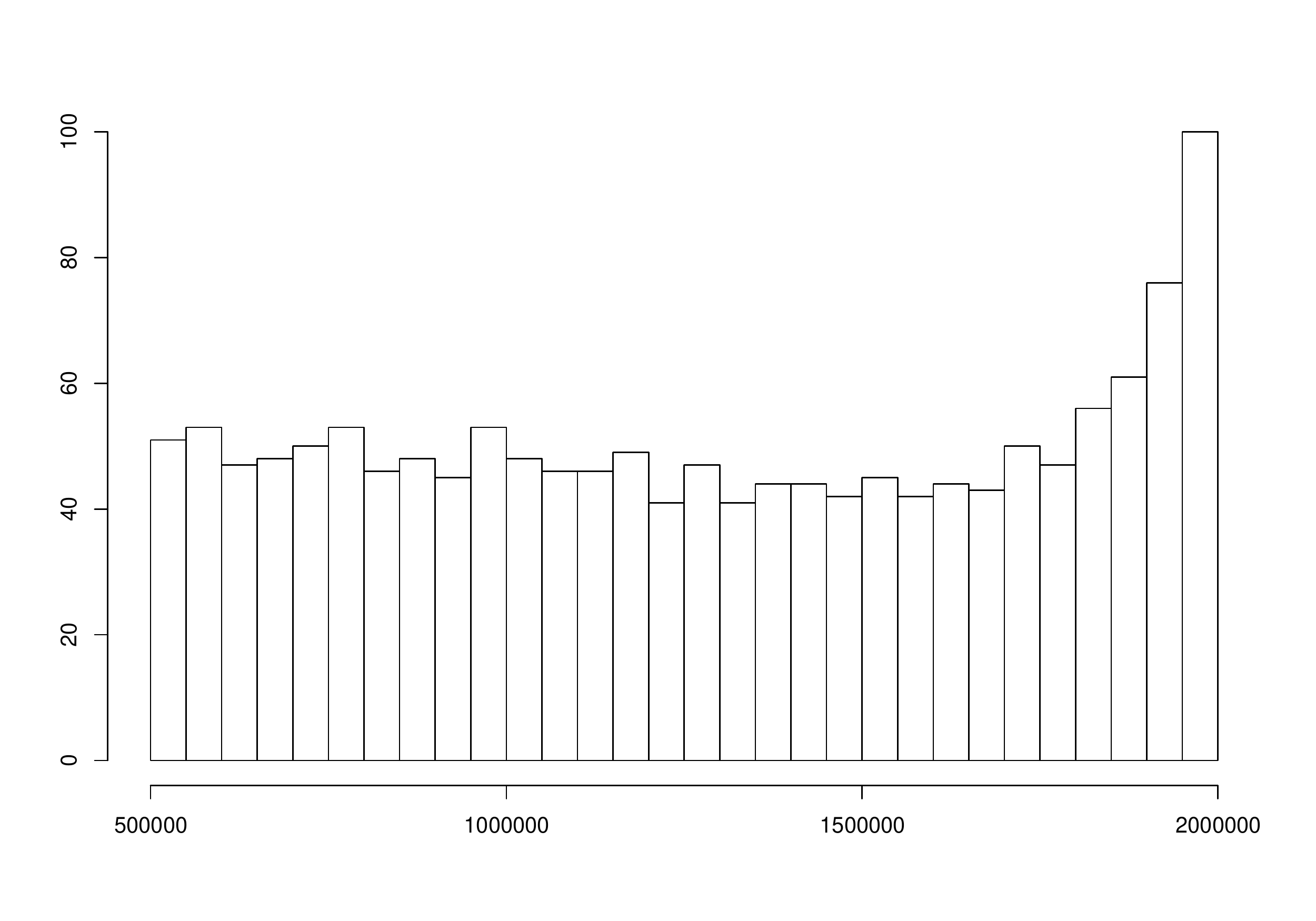}
\caption{Histogramm of grids with minimal length $100,500, 1500$ and actual length $137, 530, 1506$}\label{fig:grid-hist}
\end{figure}

Comparing an equidistant grid to the ``new'' grid we can see in Figure \ref{fig:grid-equi-vs-us} that in both algorithms, the Multinomial-Tree Algorithm and the Least-Square Monte Carlo Algorithm, we get a faster convergence using the ``new grid''. The figure shows the Value of the gas storage using an equidistant grid (star points) and the ``new'' grid (points) using different $\dt$ and different number of simulations $M$ in the example of Chapter \ref{ssec:numeric}.

\begin{figure}[htbp]
\centering
\includegraphics[trim = 15mm 10mm 5mm 10mm,angle={270}, width= 0.45\textwidth]{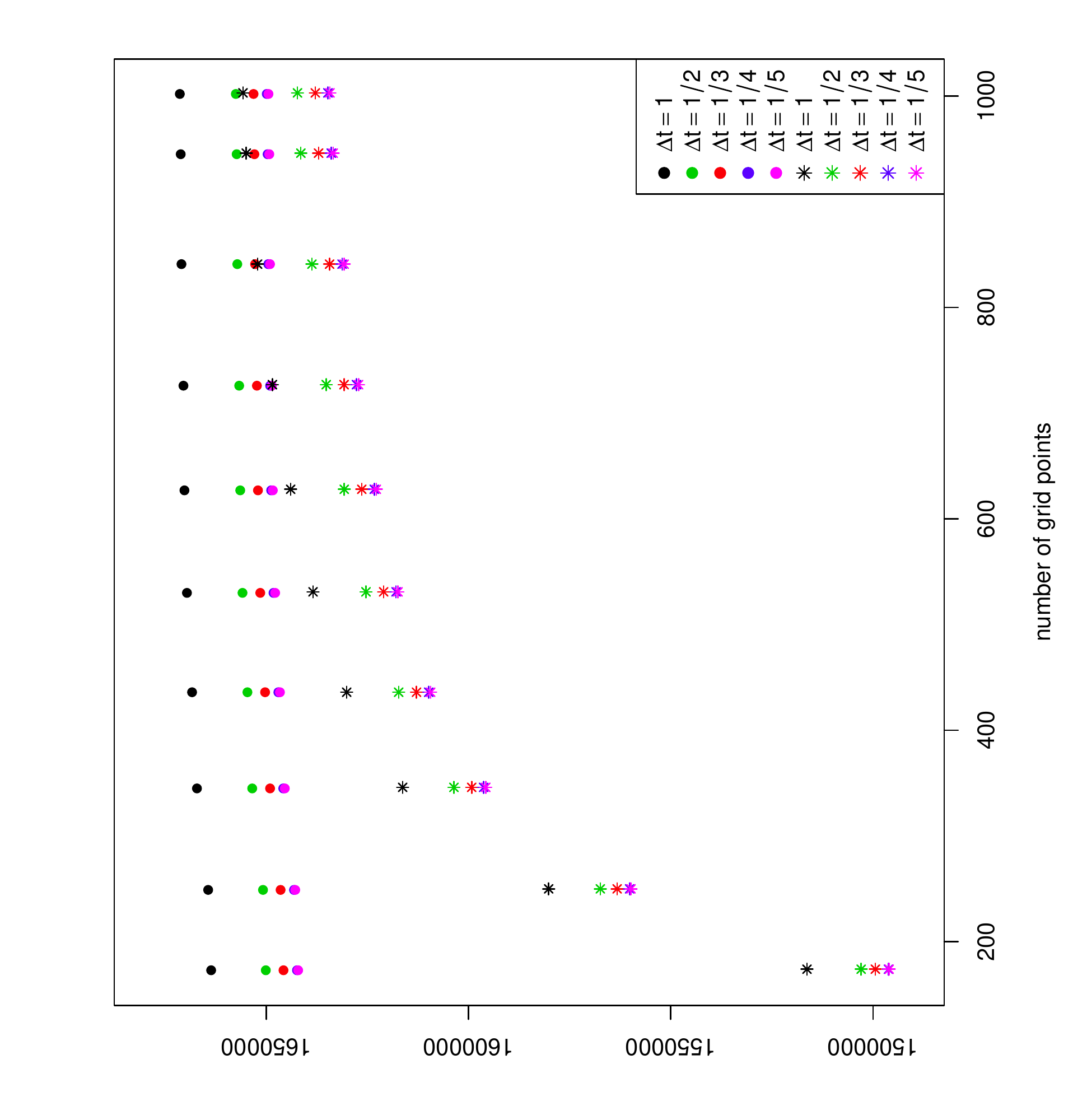}
\hspace{5mm}
\includegraphics[trim = 15mm 10mm 5mm 10mm,angle={270}, width= 0.45\textwidth]{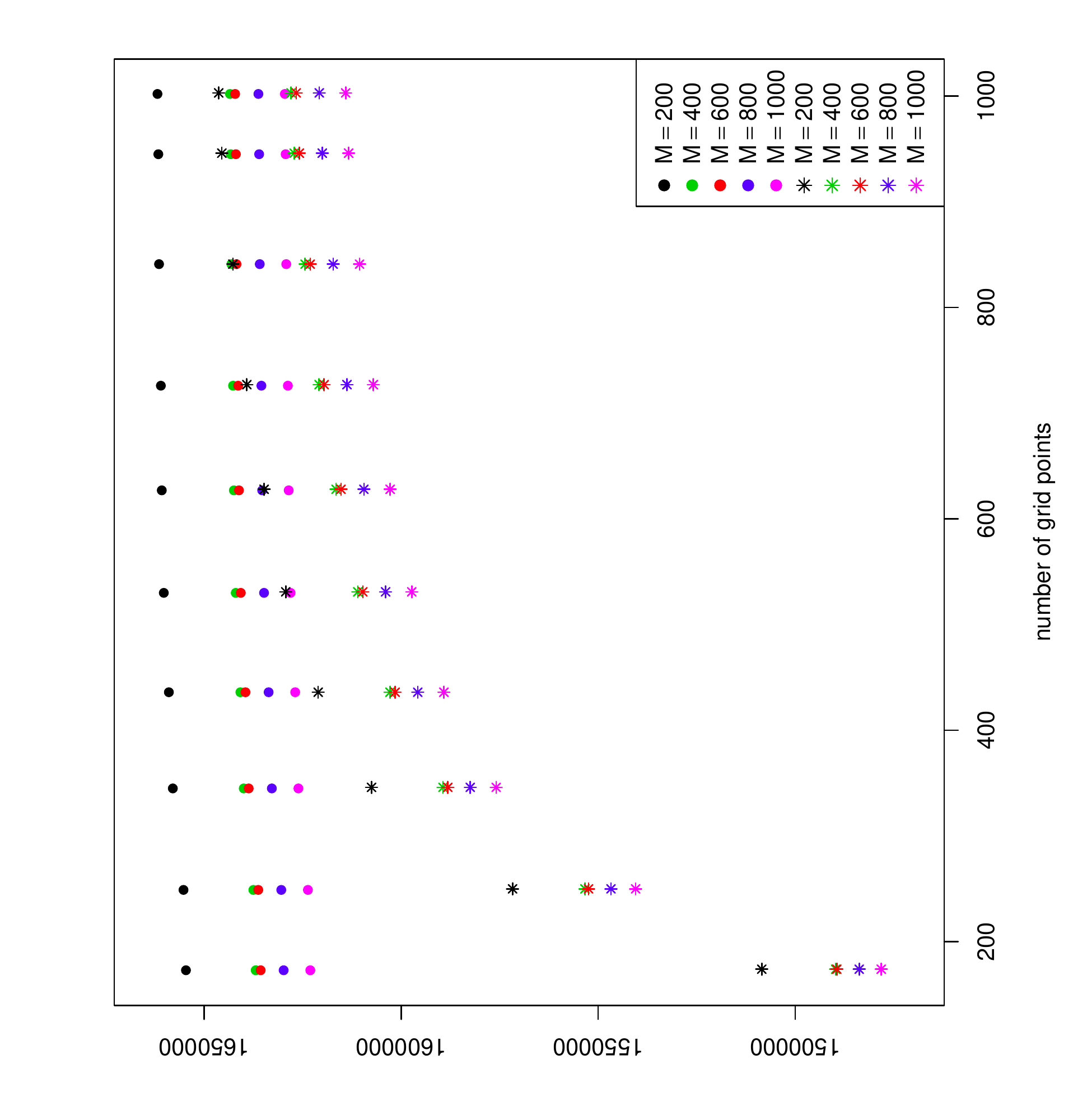}
\caption{Comparison of equidistant grid (star) and new grid (dot): Value of the gas storage using the Multinomial-Tree Algorithm (left) and the Least-Square Algorithm (right)}\label{fig:grid-equi-vs-us}
\end{figure}

\section{Numerical Analysis of the two Algorithms using an example of a gas storage facility}\label{ssec:numeric}
\subsection{An example of a gas storage facility}
First we recall that the parameters of the price process are given in Example \ref{ex:parameter-price}. The prices here are quoted in pence/therm, where $1  \text{pence/therm}=1 \text{\pounds/MMBtu}$, which will be the unit we use in the following computations. As the initial price and regime we choose $p_0=0.1 \cdot e^{\mu_1(0)} \text{\pounds/MMBtu} \approx 1.855 \text{\pounds/MMBtu}$ and $r_0=1$. The ``ask-'' and ``bid-''prices $k$ and $e$ are given by $k(p) = (1+ w_1)+z_1 = 1.01p + 0.02$ and  $k(p) = (1- w_2)-z_2 = 0.995p - 0.02$. In \cite{Sec10} the parameter $z_1$ and $z_2$ are interpreted as additional costs that result e.g. from abrasion of the pump or other technical devices. They refer to \cite{Maragos04}, who states that these costs usually won't exceed 2 cent per MMBtu. On the other hand $z_1$ and $z_2$ can be interpreted as a Bid-Ask-Spread in the market. Gas prices in \cite{dJW03} confirm that the assumption $z_1=z_2= 0.02 \pounds$ is reasonable. The choice of the parameter $w_1$ and $w_2$ can be interpreted as a loss of gas at the pump, that usually is not more than $1\%$, see \cite{Maragos04}. Details can be found in \cite{Sec10}. Another possible interpretation is to include transaction costs.

A realistic example of a gas storage contract that is based on the Stratton Ridge salt cavern in Texas can be found in \cite{thompson2009natural}. Referring to  that salt cavern we choose
$$b^\mi = 500 000 \text{ MMBtu},\quad b^\ma = 2 000 000 \text{ MMBtu}, \quad x_0 = 1 000 000 \text{ MMBtu}. $$
We choose further
$$i^\mi(x) = -70.71 \sqrt{x} , \quad i^\ma(x) = -0.032 x + 68170.$$
Thus it takes approximately 20 days to ``empty'' the storage and 78 days to ``fill'' the storage.

We choose a duration of $1$ year that corresponds to $N=250$ trading days starting from February, since the prices were fitted starting from that time. The terminal reward function is given by
$$h_N(x,p) =\begin{cases}
e(p) (x-x_0), & \text{ if } x>x_0 ,\\
0, & \text{ if } x=x_0,\\
k(p) (x-x_0)& \text{ if } x<x_0.
\end{cases}$$

\subsection{Further computations  and comparison of the two algorithms}
Before we look at further computations we investigate the influence of some of the parameters of the algorithms.
Figure \ref{fig:grid} shows the value of the gas storage using the Multinomial-Tree Algorithm (left) and the Least-Square Algorithm right with different values of $\dt$ and $M$ respectively and with different number of grid points. We see a little kink in all computation at about $530$ grid points in the grid. This is the number we use
for the grid size in the following.

\begin{figure}[htbp]
\centering
\includegraphics[trim = 20mm 10mm 5mm 10mm,angle={270}, width= 0.4\textwidth]{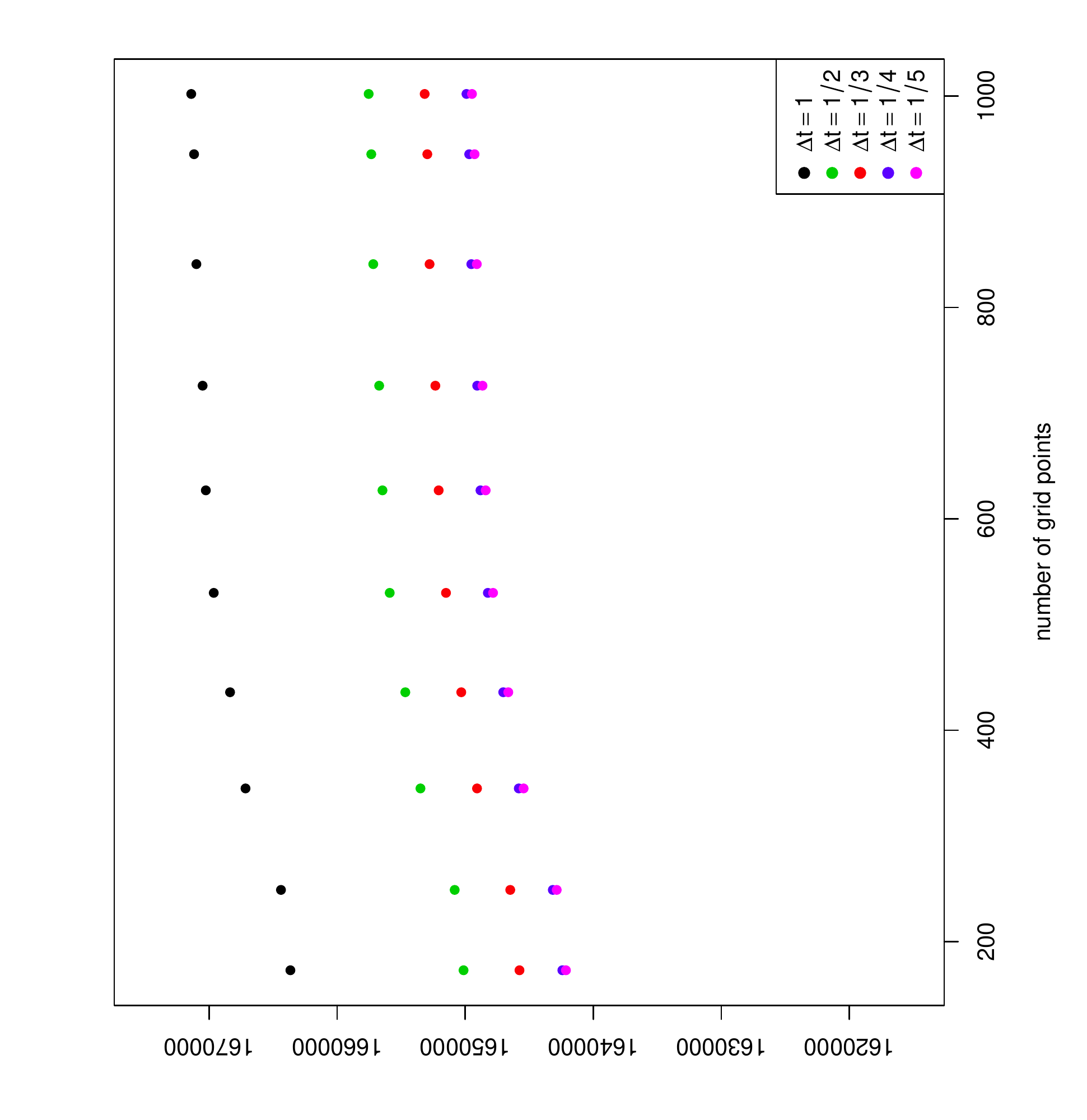}
\hspace{5mm}
\includegraphics[trim = 20mm 10mm 5mm 10mm,angle={270}, width= 0.4\textwidth]{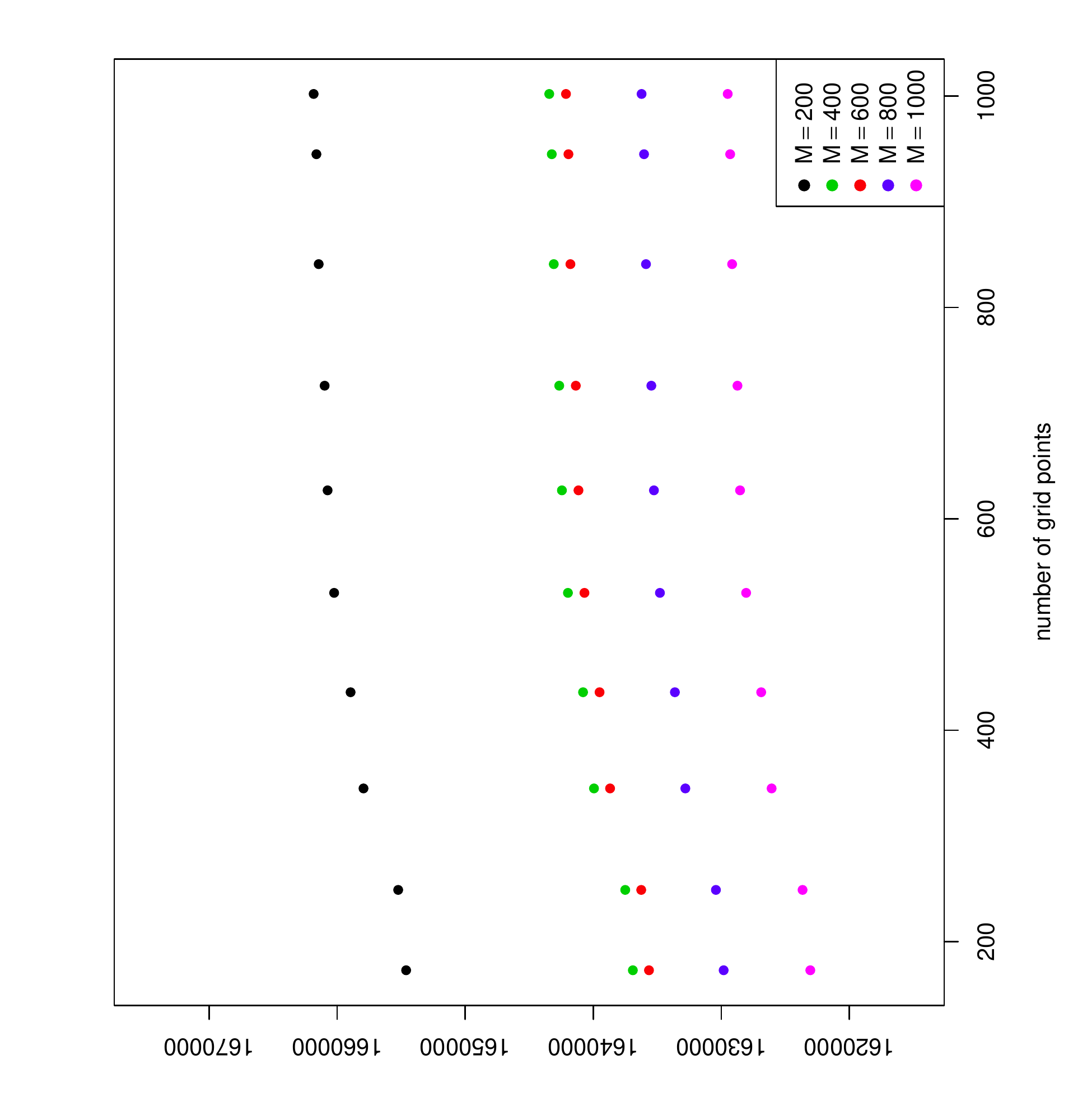}
\caption{Value of the gas storage for different numbers of grid points using the Multinomial-Tree Algorithm (left, with different $\dt$) and the Least-Square Algorithm (right, with different $M$)}\label{fig:grid}
\end{figure}

Next we want to take a closer look at the Multinomial-Tree Algorithm and see how $\dt$ influences the value of the gas storage. Table \ref{tab:tree-deltat} shows the computed values of the gas storage using different trees (cp. Figure \ref{fig:pircetreeinalgorithm}) with $\dt = 1, 1/2, 1/3, 1/4, 1/5$. Since the difference between $\dt=1/4$ and $\dt=1/5$ is rather small, we have chosen $\dt=1/4$.

\renewcommand{\arraystretch}{1.3}
\begin{table}[htbp]
\begin{tabular}{|l | c | c |c |c|c| }
\hline
$\dt$ & $1$ & $1/2$ & $1/3 $&$1/4 $& $1/5$\\ \hline
Value of the gas storage (in \pounds) & $1669631$ & $1655893$&  $1651499$& $1648229$ &$1647823$\\ \hline
\end{tabular}\\[2mm]
\caption{Value of the gas storage using the Multinomial-Tree Algorithm and different values of $\dt$.}\label{tab:tree-deltat}
\end{table}
In case of the Least-Square Algorithm we need to decide how many paths are to be simulated. Figure \ref{fig:ls-M} shows the value of the gas storage as a function of $M$ for $5$ simulation studies. As basis functions we took $1,x, x^2, x^3$. The variance of the value still seems pretty high for $M=1000$, which is due to a spike in the lowest simulation. Thus we take a further look at more simulations in Figures \ref{fig:ls-hist-M} and \ref{fig:ls-box}, that show histograms and boxplots for $100$ simulations each using $M=1000$ and $M=2000$. In the histograms we additionally plotted the means, that are $1644828$ for $M=1000$ and $1645134$ for $M=2000$. These means are also indicated in the boxplots as star points together with the standard deviations, that are $9109$ and $6516$ respectively. One can see that the difference in mean is not high, where the variance as we would expect, decays from $M=1000$ to $M=2000$ by a bigger amount.

\begin{figure}[htbp]
\centering
\includegraphics[trim = 20mm 10mm 5mm 10mm,angle={270}, width= 0.4\textwidth]{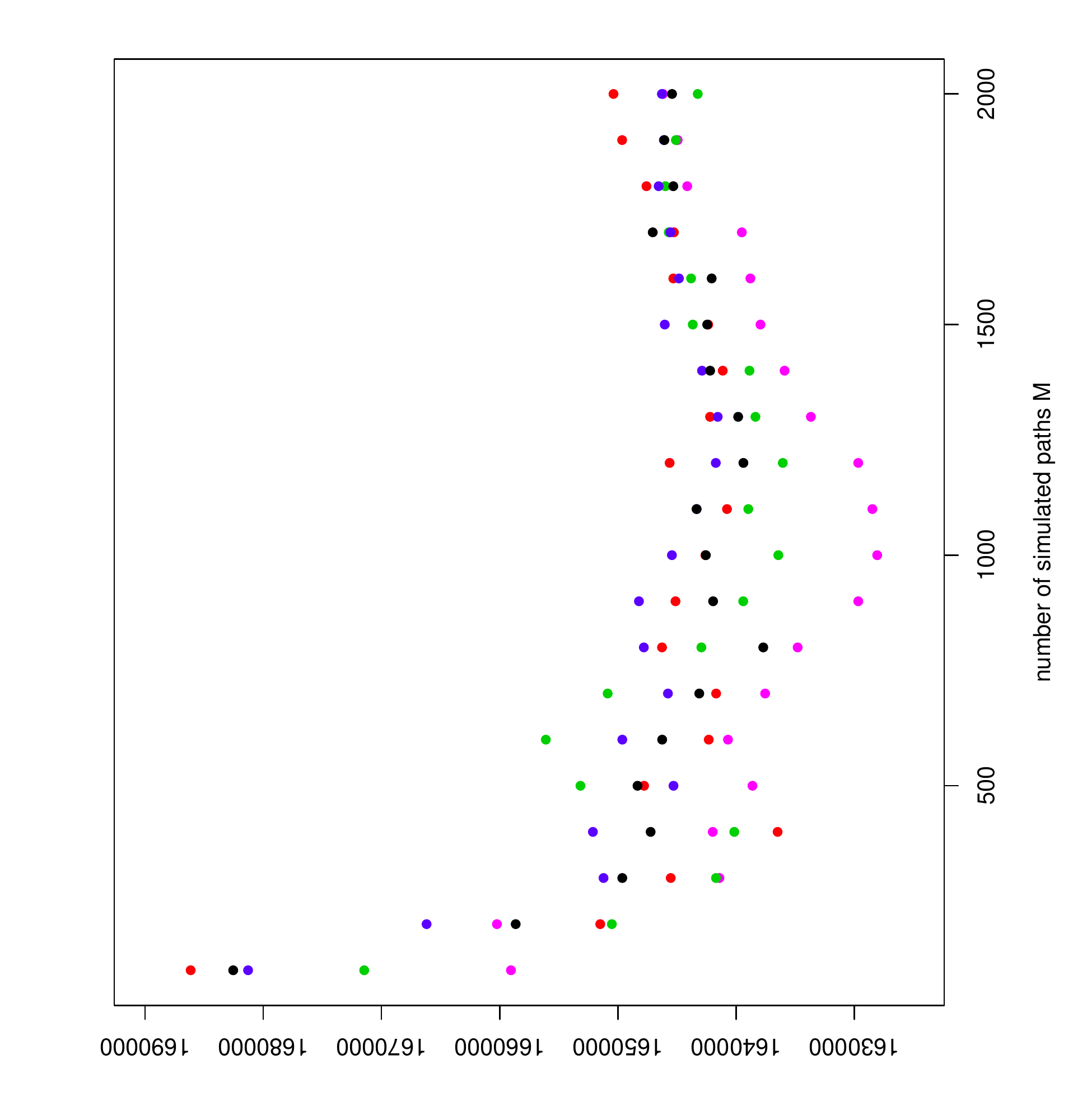}
\caption{Value of the gas storage dependent on the number simulations $M$ using the Least-Square Algorithm}\label{fig:ls-M}
\end{figure}

\begin{figure}[htbp]
\centering
\includegraphics[trim = 10mm 10mm 5mm 10mm,angle={270}, width= 0.4\textwidth]{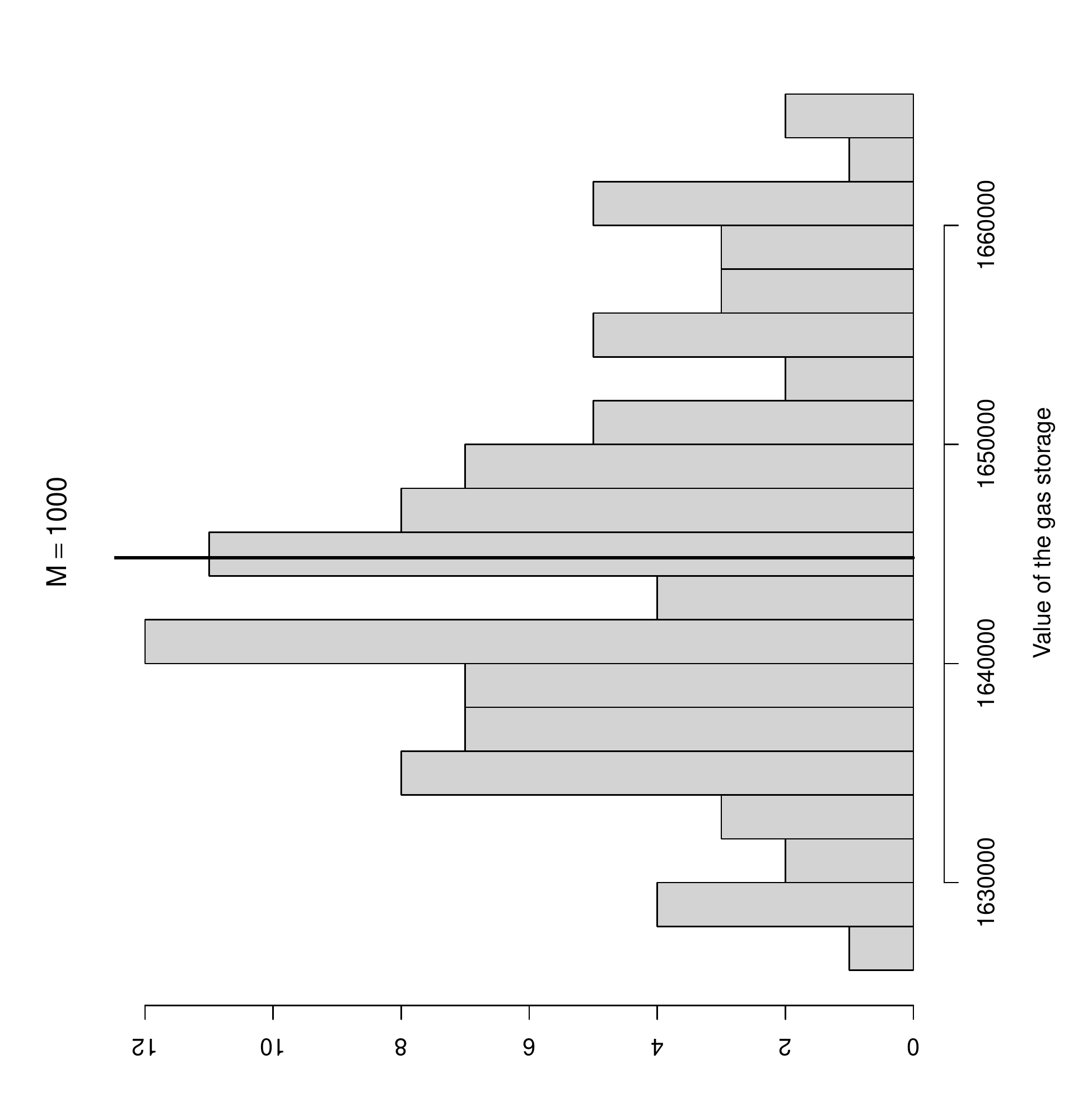}
\hspace{5mm}
\includegraphics[trim = 10mm 10mm 5mm 10mm,angle={270}, width= 0.4\textwidth]{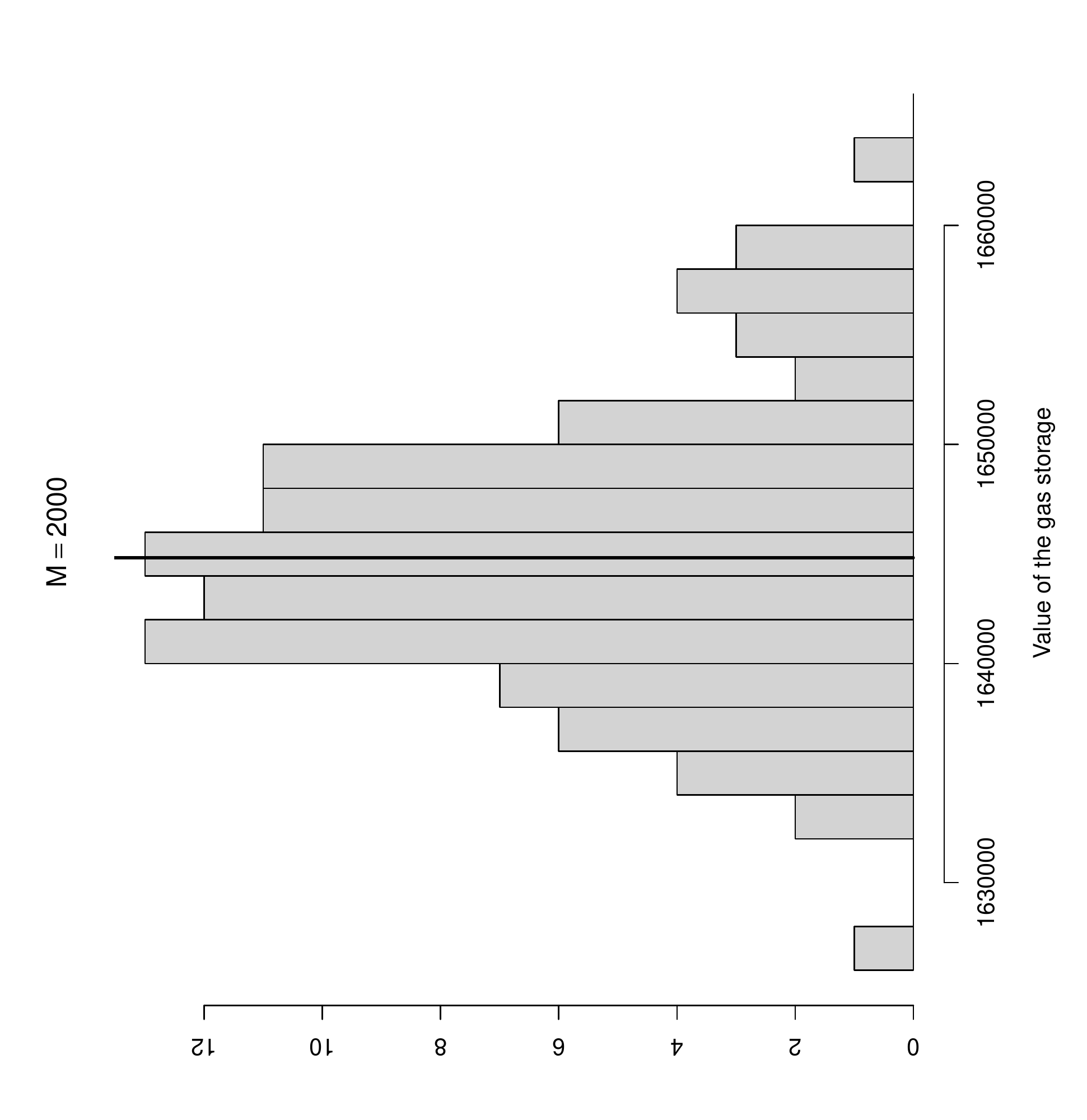}
\caption{Histogramm of the Value of the gas storage using the Least-Square Algorithm}\label{fig:ls-hist-M}
\end{figure}

\begin{figure}[htbp]
\centering
\includegraphics[trim = 20mm 10mm 5mm 10mm,angle={270}, width= 0.4\textwidth]{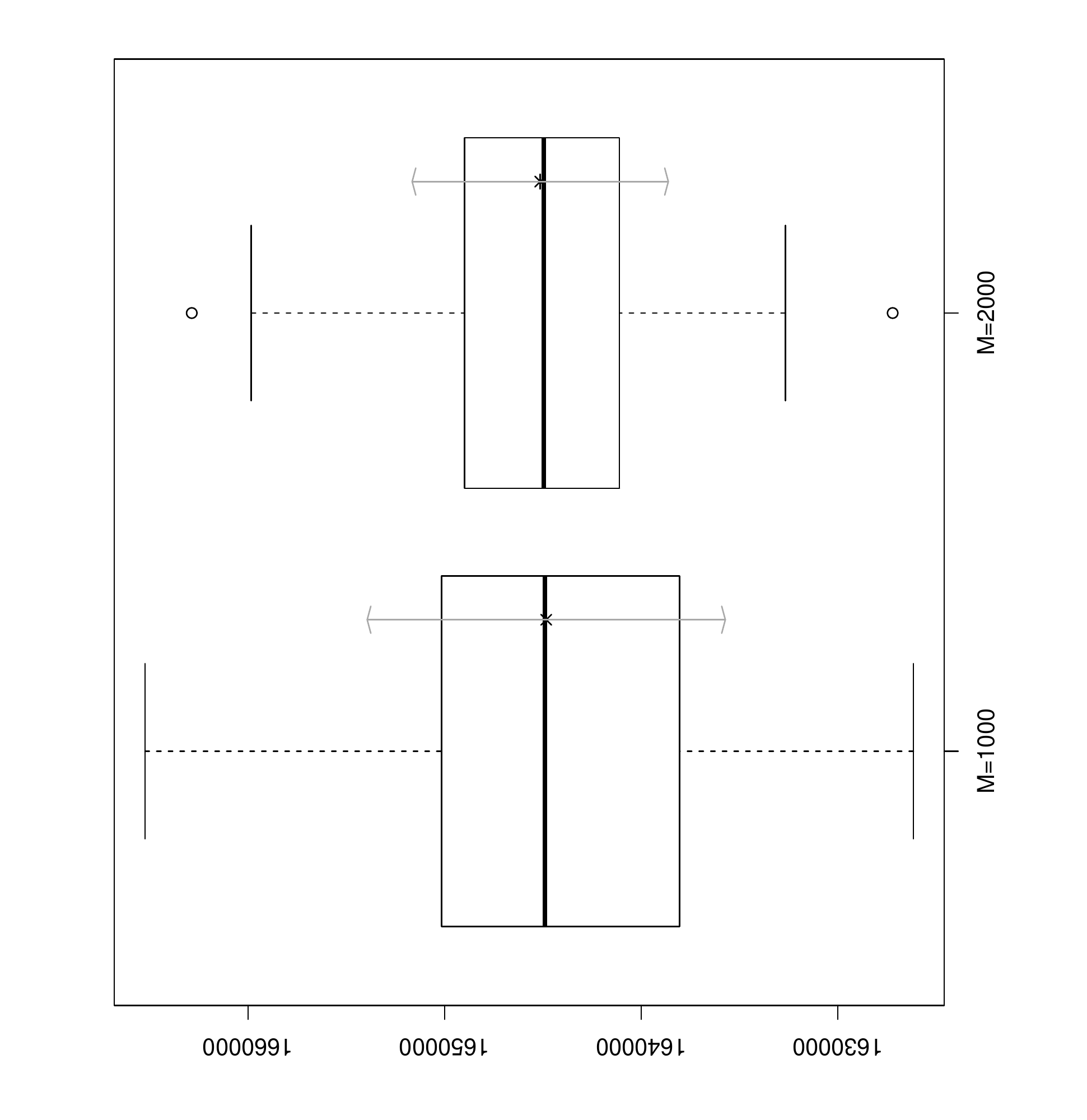}
\caption{Boxplot of the Value of the gas storage using the Least-Square Algorithm}\label{fig:ls-box}
\end{figure}

\begin{remark}[Comparison of the algorithms]
As we can see both algorithms return a value of approximately $1.64$ Million \pounds. While the Multinomial-Tree Algorithm only needs to be run once, since we are computing in the end the ``exact'' value on the grid, the Monte-Carlo Algorithm needs to be run several times to get a better impression of the estimate.  This fact makes the Multinomial-Tree Algorithm in case of runtime more interesting compared to the Least-Square Algorithm. Although the Multinomial-Tree Algorithm might be more complicated in implementation, especially when $\dt$ gets even smaller than $1/5$, once it is implemented it runs faster than the Least-Square Algorithm, e.g. because of less grid points. The advantage of the Least-Square Algorithm is clearly the applicability to a wider range of price models.
\end{remark}

To get an impression of the structure of the optimal policy we now have a closer look at the policy bounds. Figures \ref{fig:bounds-tree} and \ref{fig:bounds-ls} show the policy bounds computed by the Multinomial-Tree Algorithm with $\dt=1/4$ and the Least-Square Algorithm with $M=1000$. The lower bounds $\ul b_n(p,1)$ and $\ul b_n(p,2)$ are marked light blue and yellow, where the upper bounds $\ol b_n(p,1)$ and $\ol b_n(p,2)$ are marked blue and orange. The figure shows three pictures of all bounds, where we ``zoom in'' to the part where the bounds do not equal $b^\mi$ or $b^\ma$. The picture at the right shows just the bounds at regime $2$ in order to see more clearly that they have the same structure as those in regime $1$.
Recall that the optimal policy is to inject gas to possibly reach $\ul b$ if the current amount of gas is below $\ul b$ and if the current amount of gas is above $\ol b$ to withdraw gas if possible to reach $\ol b$, and do nothing in between.

As we can see from the figures the bounds are decreasing as the price increases, which follows the intuitive decision to sell gas if the price is high and to buy gas if the price is low. Also in many cases the bounds equal $b^\mi$ or $b^\ma$ that corresponds to the decision withdraw all gas possible and inject all gas possible, respectively. There is only a small band around the time dependent means of the price model where these bounds do not equal those specific cases. The gap, where ``do nothing'' is optimal, is caused by the fact that $k(p)$ and $e(p)$ are not equal, since those two functions make the difference in the optimization problems that need to be solved to get the policy bounds. The policy bounds at the end of the time horizon highly depend on the terminal reward function and thus have a different structure than the ones before.\\

\begin{figure}[htbp]
\centering
\includegraphics[trim = 20mm 10mm 5mm 10mm,angle={270}, width= 0.23\textwidth]{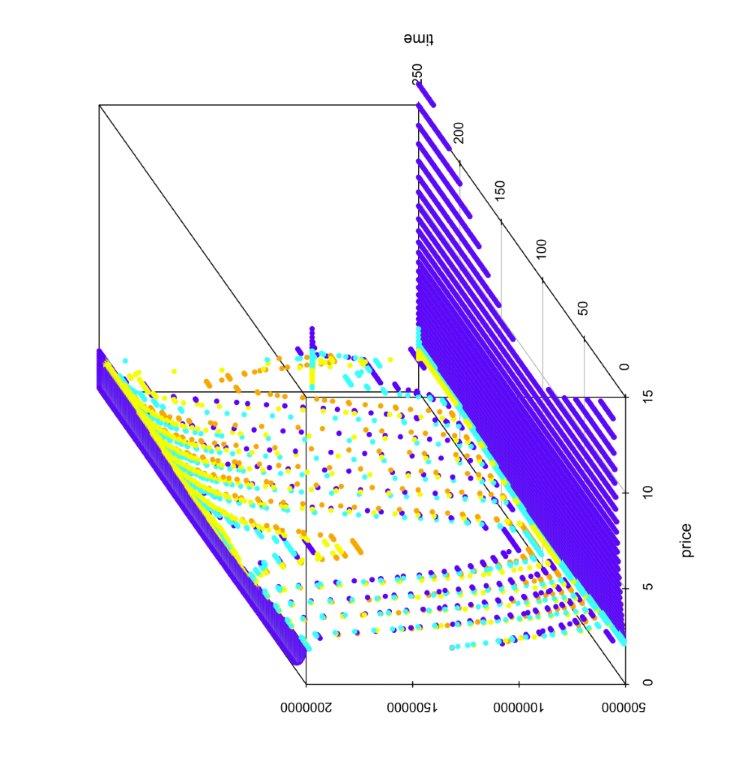}
\hspace{1mm}
\includegraphics[trim = 20mm 10mm 5mm 10mm,angle={270}, width= 0.23\textwidth]{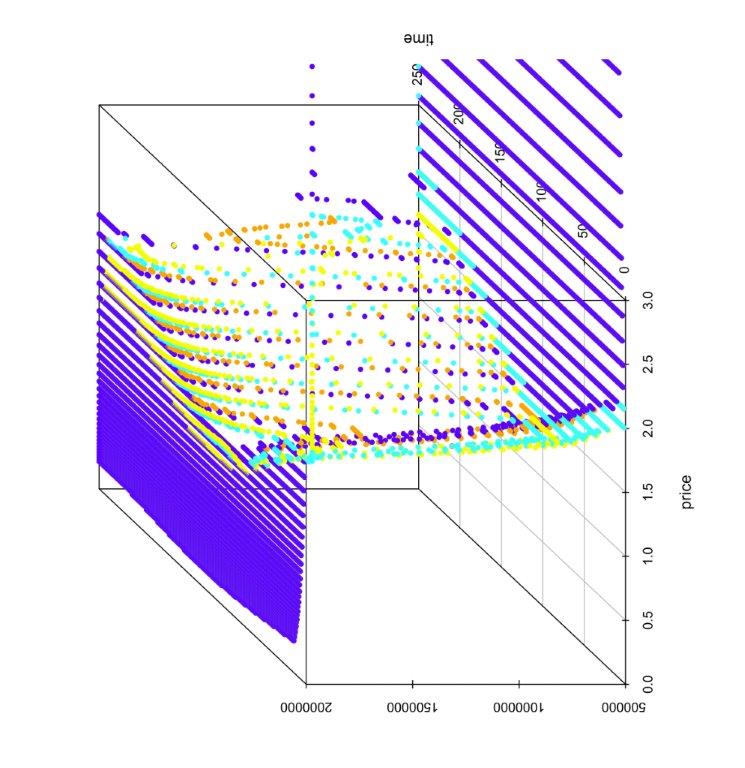}
\hspace{1mm}
\includegraphics[trim = 20mm 10mm 5mm 10mm,angle={270}, width= 0.23\textwidth]{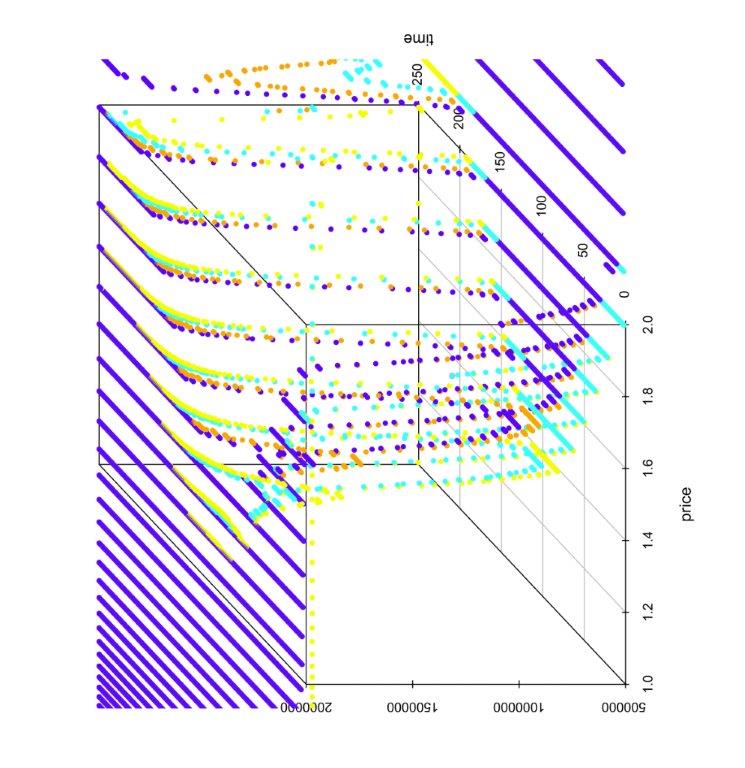}
\hspace{1mm}
\includegraphics[trim = 20mm 10mm 5mm 10mm,angle={270}, width= 0.23\textwidth]{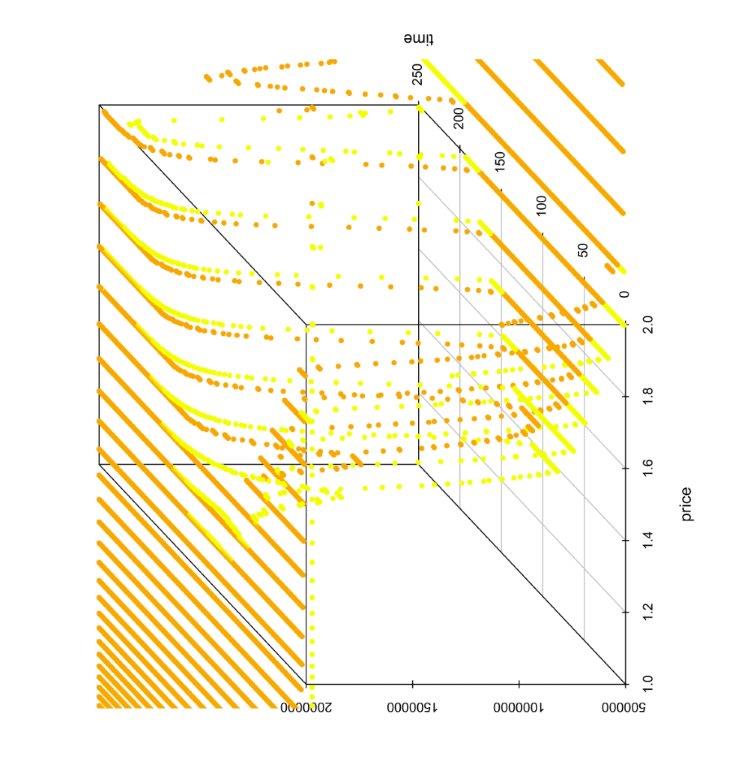}
\caption{Policy bounds using Tree Algorithm}\label{fig:bounds-tree}
\end{figure}

\begin{figure}[htbp]
\centering
\includegraphics[trim = 20mm 10mm 5mm 10mm,angle={270}, width= 0.23\textwidth]{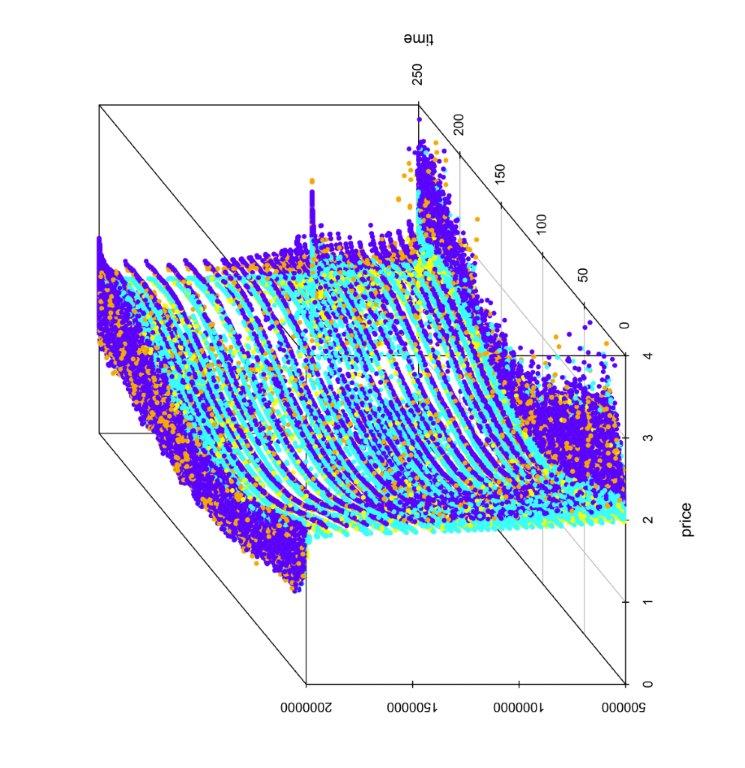}
\hspace{1mm}
\includegraphics[trim = 20mm 10mm 5mm 10mm,angle={270}, width= 0.23\textwidth]{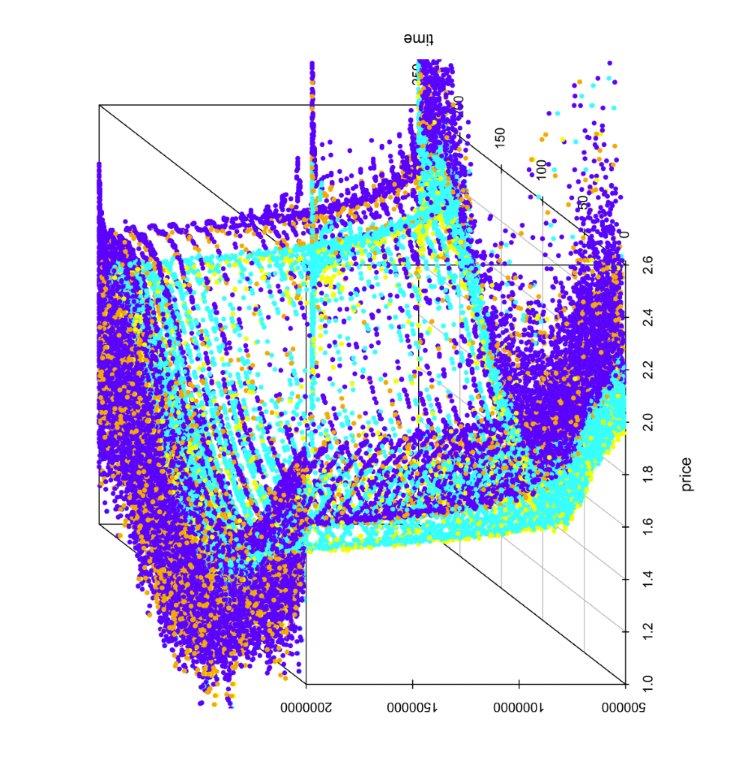}
\hspace{1mm}
\includegraphics[trim = 20mm 10mm 5mm 10mm,angle={270}, width= 0.23\textwidth]{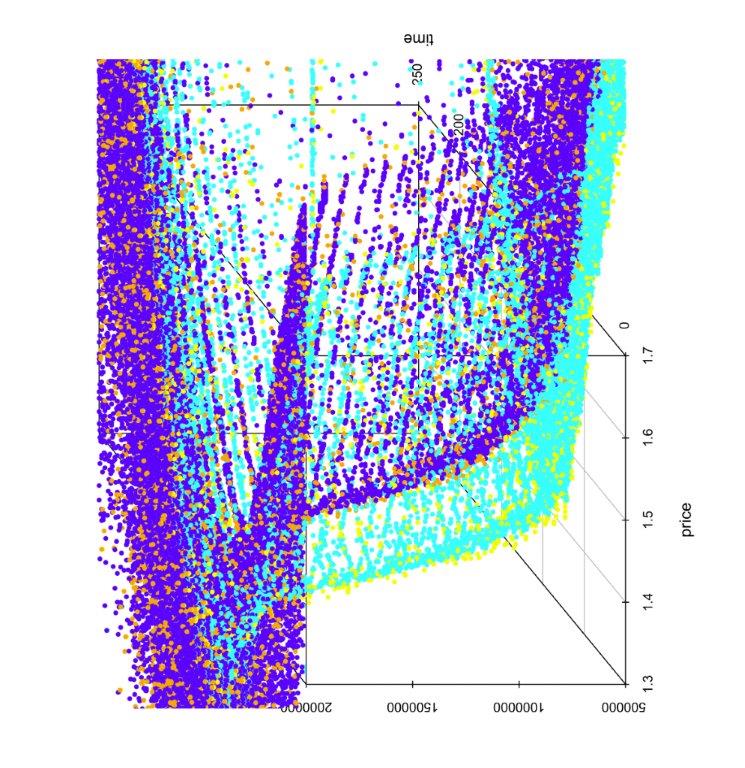}
\hspace{1mm}
\includegraphics[trim = 20mm 10mm 5mm 10mm,angle={270}, width= 0.23\textwidth]{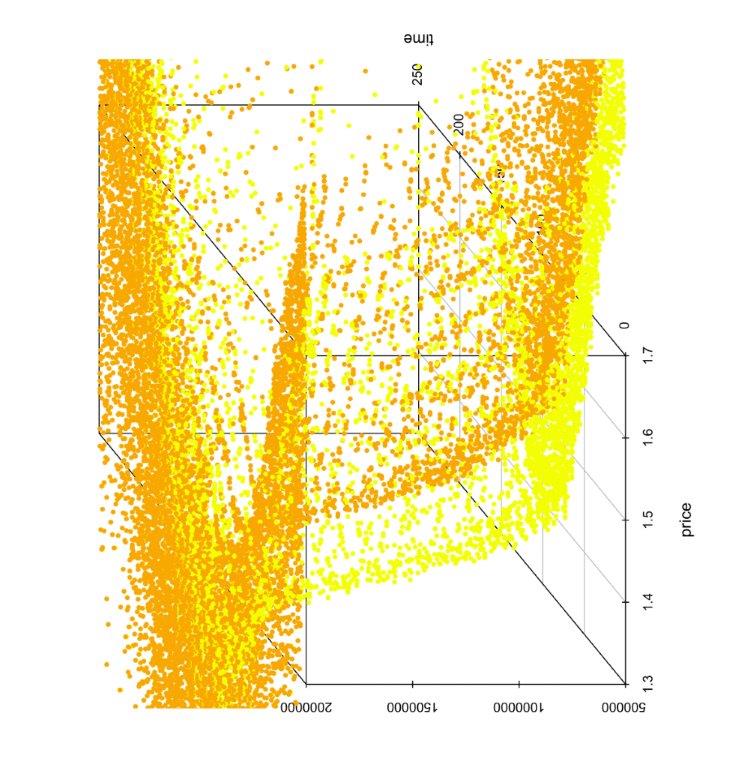}
\caption{Policy bounds using Least-Square Algorithm}\label{fig:bounds-ls}
\end{figure}

As we can guess from the figures of the policy bounds, the optimal policy is in many cases of ``bang-bang''-type in the sense that it is optimal to inject as much gas as possible, to withdraw as much gas as possible or do nothing. This policy is a rather intuitive policy. We compare the best bang-bang policy which we compute by restricting the decisions in the algorithms to the extreme ones to the optimal policy.

Looking at Table \ref{tab:value-banbang-vs-structure} we see that using the best ``bang-bang'' type policy $\pi^\circ$ leads to a lower value than using the optimal policy $\pi^\star$, but the difference is rather small. It lies below $1\%$. \\

\begin{table}[htbp]
\renewcommand{\arraystretch}{1.3}
\begin{tabular}{|l |c | c| c|}
\hline
 & value of gas storage & value of gas storage & difference\\
 method & using optimal $\pi^\star$ & using best bang-bang $\pi^\circ$ & in procent\\ \hline
Multinomial-Tree Algorithm & $1648229$ & $1637366$ & $0.660$\\ \hline
Least-Square Algorithm & $1653958$ & $1653361$ & $0.036$\\ \hline
\end{tabular}\\[2mm]
\caption{Value of the gas storage using the best bang-bang type policy $\pi^\circ$ and the optimal policy $\pi^\star$.}\label{tab:value-banbang-vs-structure}
\end{table}

Thus we look at how many times we actually decide not to inject/withdraw all gas possible. Table \ref{tab:strategy-bangbang-vs-structure} shows the absolute value in three sample paths, see Figures \ref{fig:tree-path-and-volumlevel} and \ref{fig:ls-path-and-volumlevel}. Here optimal decision $\ul b_n(p,r)$ or $\ol b_n(p,r)$ means deciding to hit those bounds in the next step and not having injected or withdrawn all possible gas. We can see that we use those actions very rarely. So in the end, if we are using a rather intuitive bang-bang-type policy the error is not that big, though we need to realize that there is a difference.

\begin{table}[htbp]
\renewcommand{\arraystretch}{1.3}
\centering
\begin{tabular}{|l|l|c|c|c|c|c|}
\cline{3-7}
\multicolumn{2}{c|}{ } & \multicolumn{5}{c|}{optimal decision} \\ \cline{1-7}
method & path & $\ul b_n(p,r)$ & $\ol b_n(p,r)$& $i^\ma$ & $i^\mi$ & $0$ \\ \hline
\multirow{3}{3cm}{Multinomial-Tree Algorithm} & red &  1 & 3 & 139 &48 &  59\\ \cline{2-7}
& green& 0  & 14 & 131 & 43 & 62 \\ \cline{2-7}
& blue &  1 &  3 &149  & 30 & 61 \\ \cline{1-7}
\multirow{3}{3cm}{Least-Square Algorithm} & red &  0& 0 & 94 & 20 & 136 \\ \cline{2-7}
& green &  1 & 5 &120 & 49 & 75\\ \cline{2-7}
& blue &  1 & 3 & 120 & 49 & 77 \\ \cline{1-7}
\end{tabular}\\[2mm]
\caption{Absolute values of the used policies in three sample paths}\label{tab:strategy-bangbang-vs-structure}
\end{table}

Finally Figures \ref{fig:tree-path-and-volumlevel} and \ref{fig:ls-path-and-volumlevel} show on the one hand three simulated gas price paths as well as the corresponding underlying Markov chain for the regimes. On the other hand the figures show the volume level over time that results from using the optimal policy computed with the Multinomial-Tree and the Monte Carlo Algorithm respectively. In total we see that, there is no general rule as ``fill in summer, empty in winter'' for the gas storage, since in the end it highly depends on the current price. Recall that time $0$ corresponds to the beginning of February. For example we see that in figure \ref{fig:tree-path-and-volumlevel} at time point about $180$ we use the high spike in the green curve to empty the storage and then fill up the storage again. Or in the blue curve we use  the rather low prices at the beginning to fill up the storage before emptying it. In figure \ref{fig:ls-path-and-volumlevel} we see that we use the rather high prices in the red curve to empty the storage really fast, where later the low prices result in a rather fast filling of the storage starting from time point $90$.

\begin{figure}[htbp]
\centering
\includegraphics[trim = 10mm 10mm 5mm 10mm,angle={270}, width=0.7\textwidth]{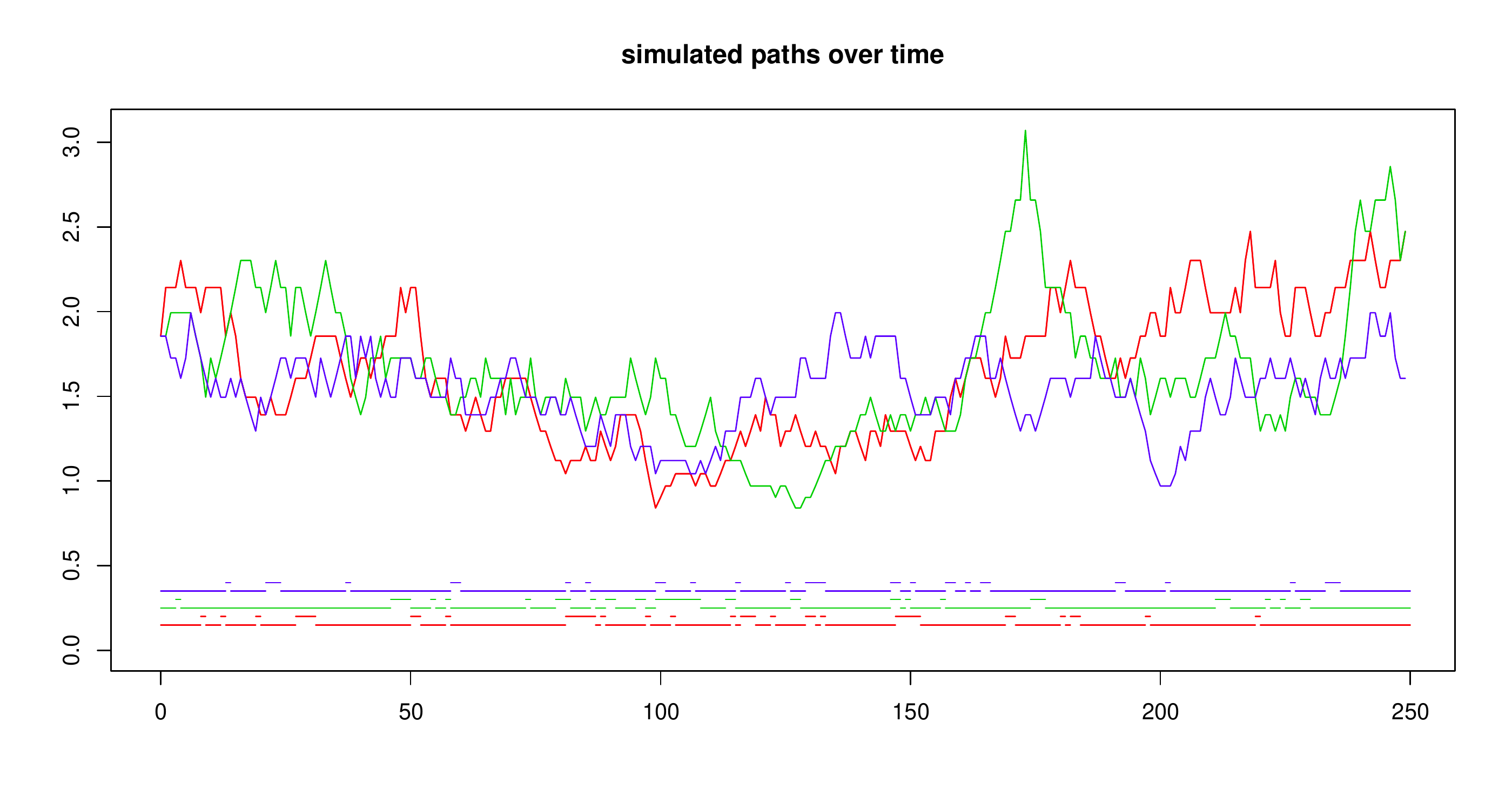}

\includegraphics[trim = 15mm 10mm 15mm 10mm,angle={270}, width=0.7\textwidth]{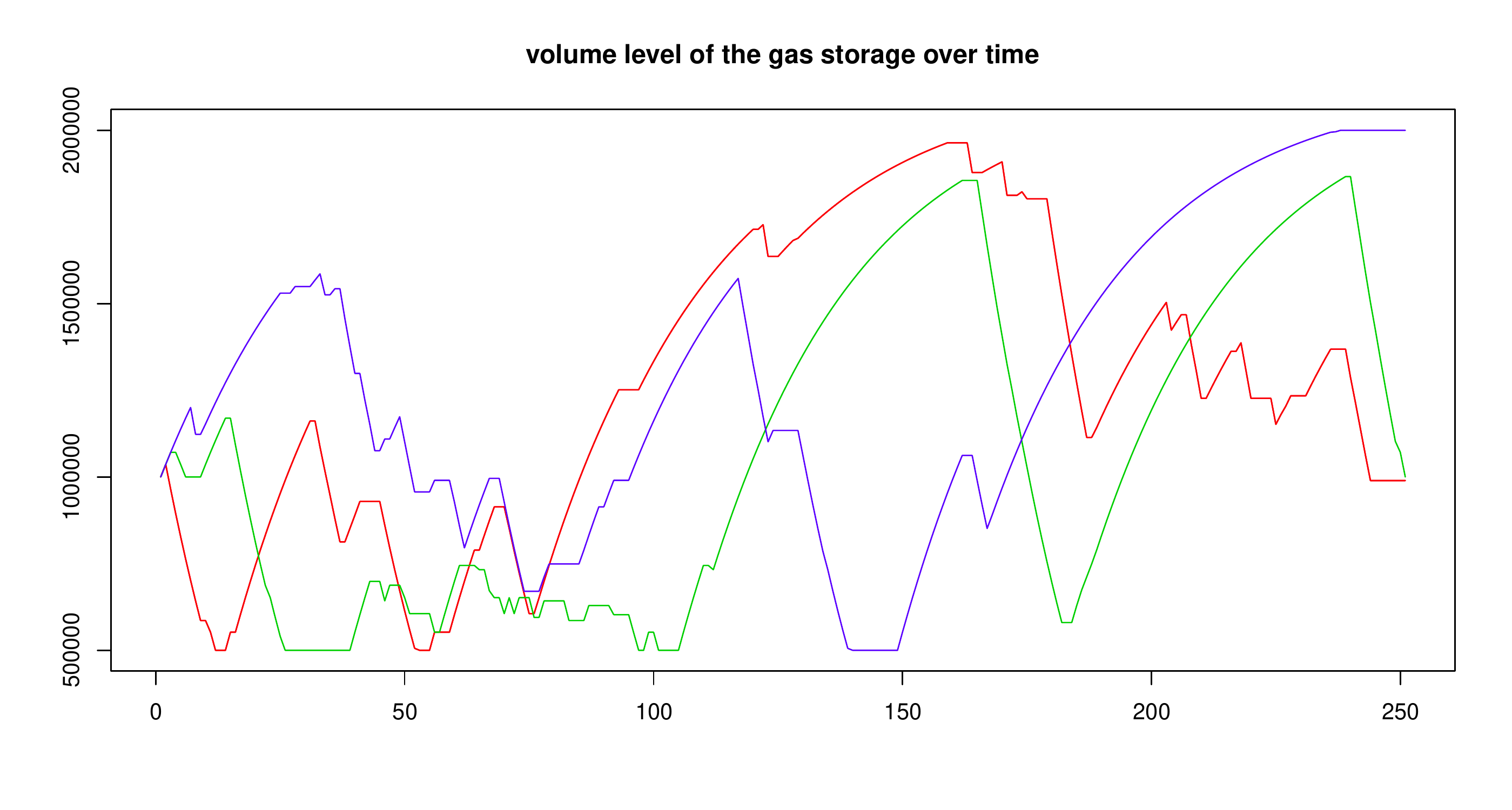}
\caption{Three simulated paths and the corresponding volume level using the optimal policy in the Multinomial-Tree Algorithm}\label{fig:tree-path-and-volumlevel}
\end{figure}

\begin{figure}[htbp]
\centering
\includegraphics[trim = 10mm 10mm 5mm 10mm,angle={270}, width=0.7\textwidth]{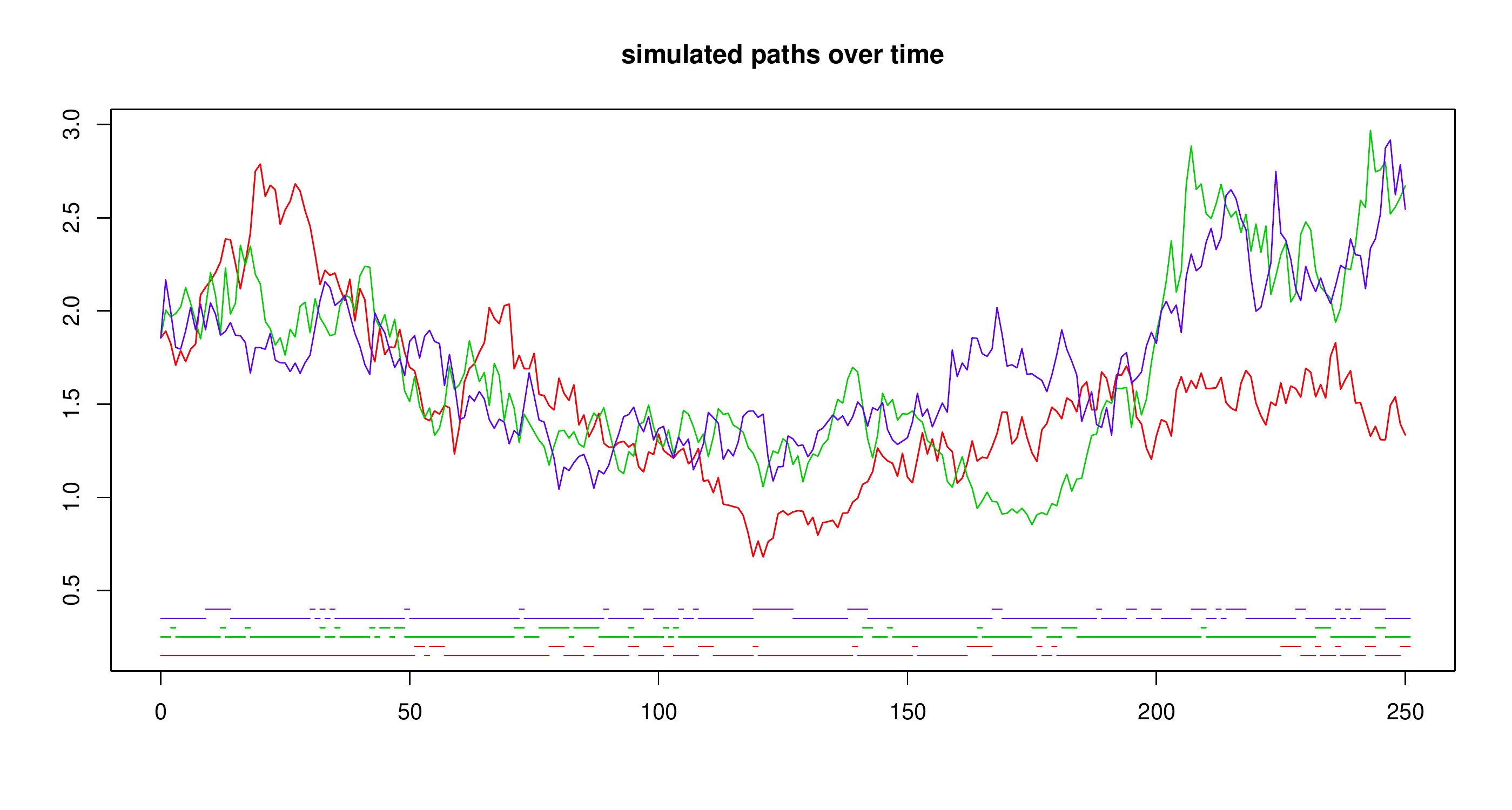}

\includegraphics[trim = 15mm 10mm 15mm 10mm,angle={270}, width=0.7\textwidth]{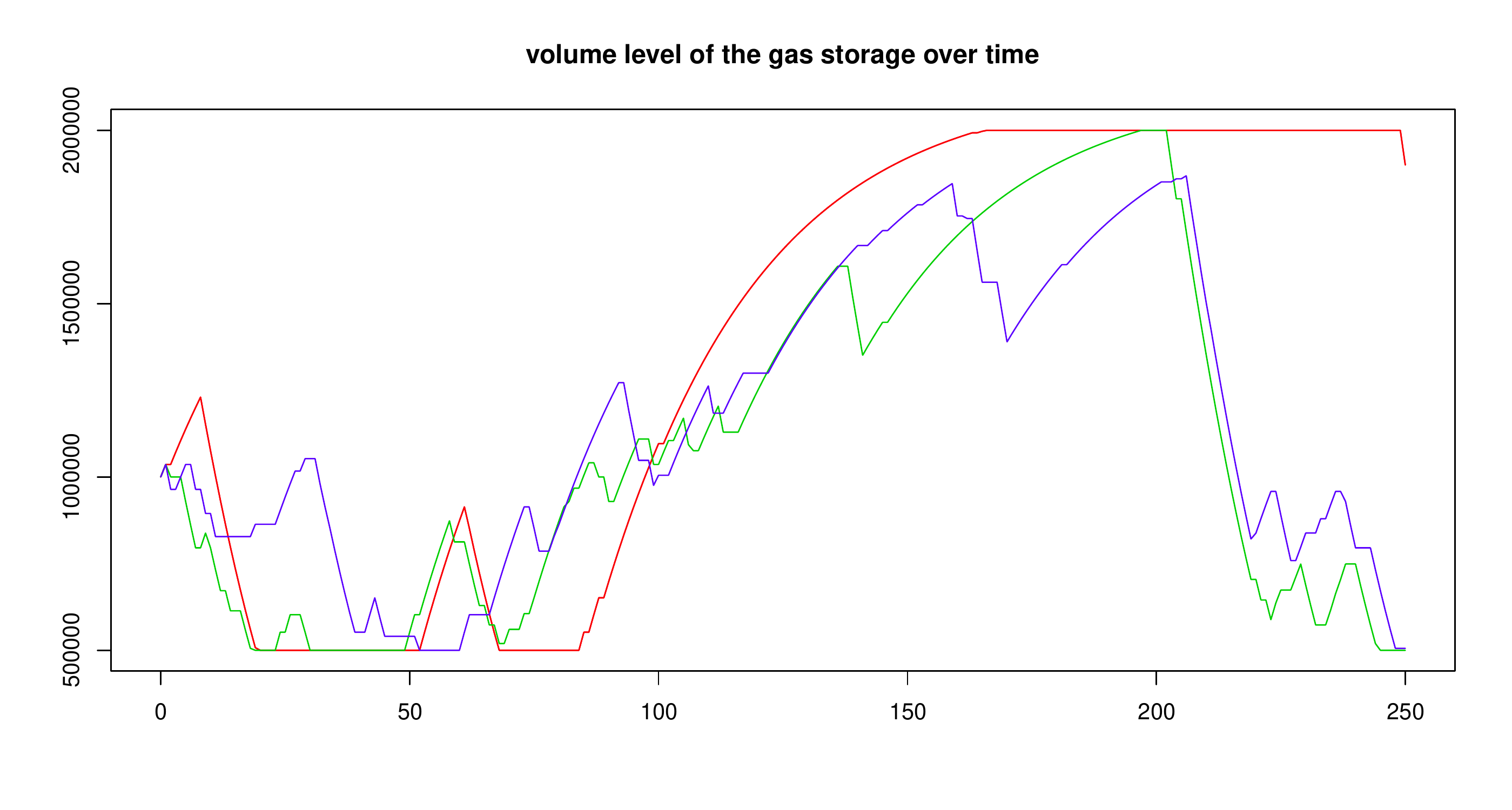}
\caption{Three simulated paths and the corresponding volume level using the optimal policy in the Least-Square Algorithm}\label{fig:ls-path-and-volumlevel}
\end{figure}

\section{Conclusion}
As we have seen, the result from \cite{Sec10} concerning the structure of the optimal policy can be extended to price processes with regime switching and storage-dependent injection and withdrawal rates. Both algorithms, the Multinomial-Tree Algorithm and the Least-Square Algorithm, work quite well for the price models with regime switching. Hints on how to choose the critical parameters of the algorithm have been given. In particular a clever choice of the grid of the gas storage level has been shown to be important. On the other hand we also conclude that instead of maximizing over all feasible policies it is almost optimal to maximize over bang-bang policies only.

Further research needs to be done in order to extend the presented algorithms to price models with jumps. While it should not be a problem to include jumps in the Least-Square Algorithm, since the process only needs to be simulated in the end, it would be interesting to find out if there is a way to approximate the price model by a recombing tree in a similar way. Further \cite{boogert2011gas} and \cite{BdJ08} present some simulation studies on the choice of the basis functions in the Least-Square Algorithm. Their result that the choice of basis functions does not have a big effect on the value of the algorithm needs to be verified in the Least-Square Algorithm with regime switching. Here it would also be interesting to investigate in what extend a regression in $(p,r)$ has an effect on the storage value compared to the separate regressions in each regime. \\

\section{Acknowledgement}
The authors would like to thank Alfred M\"uller for helpful comments and suggestions on the PhD thesis of Viola Riess which also entered the paper.

\bibliographystyle{abbrv}

\end{document}